\documentclass[preprint]{elsarticle}

\makeatletter
\def\ps@pprintTitle{%
 \let\@oddhead\@empty
 \let\@evenhead\@empty
 \def\@oddfoot{}%
 \let\@evenfoot\@oddfoot}
\makeatother

\usepackage{hyperref}
\usepackage{amsmath,amssymb,amsthm}
\usepackage{hyperref}
\usepackage{amsthm}
\usepackage{graphicx,subfigure}
\usepackage{verbatim}
\usepackage{upgreek}
\usepackage{color}
\usepackage{xcolor}
\usepackage{algorithmic}
\usepackage[ruled,vlined]{algorithm2e}
\usepackage{caption}

\journal{Journal of \LaTeX\ Templates}

\newtheorem{theorem}{Theorem}
\newtheorem{lemma}[theorem]{Lemma}
\newtheorem{property}{Property}

\newcommand{\fpp}[1]{$#1$-PPT}
\newcommand{\fppExt}[1]{precedence proper $#1$-thin}
\newcommand{\fppPar}[1]{$pre$-$pthin(#1)$}

\newcommand{\fp}[1]{$#1$-PT}
\newcommand{\fpExt}[1]{precedence $#1$-thin}
\newcommand{\fpPar}[1]{$pre$-$thin(#1)$}

\newcommand{\NP}{\textsf{NP}}

\SetKwBlock{addEdgesFromPQTree}{procedure addEdgesFromPQTree($G$, $D$,  $T$)}{}
\SetKwBlock{Kpt}{function partitioned-\fp{k}($G$, $k$,  $\mathcal{V}$)}{}
\SetKwBlock{Kppt}{function partitioned-\fpp{k}($G$, $k$, $\mathcal{V}$)}{}
\SetKwBlock{KpptF}{function Partitioned-\fpp{k}($G$, $\mathcal{V}$)}{}
\SetKwBlock{KpptFC}{function Partitioned-\fpp{k}($G$, $\mathcal{V}$)}{}
\SetKwFunction{AddEdgesFromPQTree}{addEdgesFromPQTree}
\SetKw{Each}{each}
\SetKw{TRUE}{TRUE}
\SetKw{FALSE}{FALSE}
\SetKw{Break}{break}
\SetKw{Not}{not}

\DeclareMathOperator{\thin}{thin}
\DeclareMathOperator{\pthin}{pthin}

\newcommand{\mC}{{\mathcal C}}

\begin{document}

\begin{frontmatter}

\title{Precedence thinness in graphs}





\author[UBA]{Flavia~Bonomo-Braberman}
\ead{fbonomo@dc.uba.ar}

\author[UERJ]{Fabiano~S.~Oliveira}
\ead{fabiano.oliveira@ime.uerj.br}

\author[UFRJ]{Moys\'es~S.~Sampaio~Jr.\corref{mycorrespondingauthor}}
\cortext[mycorrespondingauthor]{Corresponding author}

\ead{moysessj@cos.ufrj.br}

\author[UERJ,UFRJ]{Jayme~L.~Szwarcfiter}
\ead{jayme@nce.ufrj.br}

\address[UBA]{Universidad de Buenos Aires. Facultad de Ciencias Exactas y Naturales. Departamento de Computaci\'on. Buenos Aires,
Argentina. / CONICET-Universidad de Buenos Aires. Instituto de
Investigaci\'on en Ciencias de la Computaci\'on (ICC). Buenos
Aires, Argentina.}
\address[UERJ]{Universidade do Estado do Rio de Janeiro, Brazil}
\address[UFRJ]{Universidade Federal do Rio de Janeiro, Brazil}


\begin{abstract}
Interval and proper interval graphs are very well-known graph classes, for which there is a wide literature. As a consequence, some generalizations of interval graphs have been proposed, in which graphs in general are expressed in terms of $k$ interval graphs, by splitting the graph in some special way.

As a recent example of such an approach, the classes of $k$-thin  and proper $k$-thin graphs have been introduced generalizing interval and proper interval graphs, respectively. 
The complexity of the recognition of each of these classes is still open, even for fixed $k \geq 2$. 

In this work, we introduce a subclass of  $k$-thin graphs (resp. proper $k$-thin graphs), called \fpExt{k} graphs (resp. \fppExt{k} graphs).  
Concerning partitioned \fpExt{k} graphs, we present a polynomial time recognition algorithm based on $PQ$ trees. With respect to partitioned \fppExt{k} graphs, we prove that the related recognition problem is \NP-complete for an arbitrary $k$ and polynomial-time solvable when $k$ is fixed. Moreover, we present a characterization for these classes based on threshold graphs.
\end{abstract}


\end{frontmatter}


\section{Introduction}\label{chpt:introduction}

The class of $k$-thin graphs has recently been introduced by Mannino, Oriolo, Ricci and Chandran in~\cite{Carl07} as a generalization of interval graphs. Motivated by this work, Bonomo and de Estrada~\cite{Bono18} defined the class of proper $k$-thin graphs, which generalizes proper interval graphs. A \emph{$k$-thin graph} $G$ is a graph for which there is a $k$-partition $(V_1,V_2 \ldots, V_k)$, and an ordering $s$ of $V(G)$ such that, for any triple $(p,q,r)$ of $V(G)$ ordered according to $s$, if $p$ and $q$ are in a same part $V_i$ and $(p,r) \in E(G)$, then $(q,r)\in E(G)$. Such an ordering and partition are said to be \emph{consistent}.
A graph $G$ is called a \emph{proper $k$-thin graph} if $V(G)$ admits a $k$-partition $(V_1, \ldots, V_k)$, and an ordering $s$ of $V(G)$ such that both $s$ and its reversal are consistent with the partition $(V_1, \ldots,V_k)$. An ordering of this type is said to be a \emph{strongly consistent} ordering.
The interest on the study of both these classes  comes from the fact that some \NP-complete problems can be solved in polynomial time when the input graphs belong to them~\cite{Carl07,Bono18,B-M-O-thin-tcs}. Some of those efficient solutions have been exploited to solve real world problems as presented in~\cite{Carl07}.

On a theoretical perspective, defining  general graphs in terms of the concept of interval graphs has been of recurring interest in the literature. Firstly, note that these concepts measure ``how far'' a given graph $G$ is from being an interval graph, or yet, how $G$ can be ``divided'' into interval graphs, mutually bonded by a vertex order property. 
Namely, the vertices can be both partitioned and ordered in such a way that, for every part $V'$ of the partition and every vertex $v$ of the ordering, the vertex ordering obtained from the original by the removal of all vertices except from $v$ and those in $V'$ that precede $v$, no matter which part $v$ belongs, is a canonical ordering of an interval graph.
Characterizing general graphs in terms of the concept of interval graphs, or proper interval graphs, is not new. A motivation for such an approach is that the class of interval graphs is well-known, having several hundreds of research studies on an array of different problems on the class, and formulating general graphs as a function of interval graphs is a way to extend those studies to general graphs. Given that, both $k$-thin and proper $k$-thin graphs are new generalizations of this kind. 
The complexity of recognizing whether a graph is $k$-thin, or proper $k$-thin, is still an open problem even for a fixed $k \geq 2$. For a given vertex ordering, there are  polynomial time algorithms that compute a partition into a minimum number of classes  for which the ordering is consistent (resp. strongly consistent)~\cite{Bono18,B-M-O-thin-tcs}. On the other hand, given a vertex partition, the problem of deciding the existence of a vertex ordering which is consistent (resp. strongly consistent) with that partition is \NP-complete~\cite{Bono18}.

Other generalizations of interval graphs have been proposed. As examples, we may cite the $k$-interval and $k$-track interval graphs. A \emph{$k$-interval} is the union of $k$ disjoint intervals on the real line. A \emph{$k$-interval graph} is the intersection graph of a  family of $k$-intervals. Therefore, the $k$-interval graphs generalize the concept of interval graphs by allowing a vertex to be associated with a set of disjoint intervals. The \emph{interval number} $i(G)$ of $G$~\cite{Trotter1979OnDA} is the smallest number $k$ for which $G$ has a $k$-interval model. Clearly, interval graphs are the graphs with $i(G) = 1$. A \emph{$k$-track interval} is the union of $k$ disjoint intervals distributed in $k$ parallel lines, where each interval belongs to a distinct line. Those lines are called \emph{tracks}.
A \emph{$k$-track interval graph} is the intersection graph of $k$-track intervals. The \emph{multitrack number} $t(G)$ of $G$~\cite{Kumar94,Gyarfas95} is the minimum $k$ such that $G$ is a $k$-track interval graph. Interval graphs are equivalent to the $1$-interval graphs and $1$-track interval graphs. 
The problems of recognizing $k$-interval and $k$-track interval graphs are both \NP-complete~\cite{West84,Jiang2013}, for every $k \geq 2$. 

In this paper, we define subclasses of $k$-thin and proper $k$-thin graphs called \emph{\fpExt{k}} and \emph{\fppExt{k}} graphs, respectively, by adding the requirement that the vertices of each class have to be consecutive in the order. In both cases, when the vertex order is given, it can be proved that a greedy algorithm can be used to find a (strongly) consistent partition into consecutive sets with minimum number of parts. When, instead, the partition into $k$ parts is given, the problem turns out to be more interesting. We will call \emph{partitioned \fpExt{k} graphs} (resp. \emph{partitioned \fppExt{k} graphs}) the graphs for which there exists a (strongly) consistent vertex ordering in which the vertices of each part are consecutive for the given partition. Concerning partitioned \fpExt{k} graphs, we present a polynomial time recognition algorithm. With respect to partitioned \fppExt{k} graphs, we provide a proof of \NP-completeness for arbitrary $k$, and a polynomial time algorithm when $k$ is fixed. Also, we provide a characterization for both classes based on threshold graphs.


This work is organized as follows. 
Section~\ref{chpt:preliminaries} introduces the basic concepts and terminology employed throughout the text. 
Section~\ref{chpt:pthinness} presents a polynomial time recognition algorithm for partitioned \fpExt{k} graphs.  Section~\ref{chpt:ppthinness} proves the \NP-completeness for the recognition problem  of partitioned \fppExt{k} graphs. Moreover, it proves that if the number of parts of the partition is fixed, then the  problem is solvable in polynomial time. In Section~\ref{chpt:characterization}, we describe a characterization for \fpExt{k} and \fppExt{k} graphs. Finally,  Section~\ref{chpt:future} presents some concluding remarks.
\unskip
\section{Preliminaries}\label{chpt:preliminaries}

All graphs in this work are finite and have no loops
or multiple edges. Let $G$ be a graph. Denote by $V(G)$ its vertex set, by $E(G)$ its edge set.
Denote the \emph{size} of a set $S$ by $|S|$. Unless stated otherwise, $|V(G)| = n$ and $|E(G)| = m$. 
Let $u,v \in V(G)$, define $u$ and $v$ as \emph{adjacent} if $(u,v) \in E(G)$.

Let $V' \subseteq V(G)$. The \emph{induced subgraph} of $G$ by $V'$, denoted by $G[V']$, is the graph $G[V']= (V',E')$, where  $E' =  \{ (u,v)\in E(G) \mid u, v \in V' \}$. Analogously, for some $E' \subseteq E(G)$, let  the \emph{induced subgraph} of $G$ by $E'$, denoted by $G[E']$, be the graph $G[E']= (V',E')$, where  $V' =  \{ u, v \in V(G) \mid  (u,v)\in E' \}$. 
A graph $G'$ obtained from $G$ by \emph{removing} the vertex $v \in V(G)$ is defined as $G' = (V(G) \setminus \{v\},  E')$, where $E'=\{(u,w) \in E \mid u \neq v \text{ and } w \neq v\}$ 

Let $v \in V(G)$, denote by $N(v)=  \{u \in V(G) \mid v \text{ and } u \text{ are adjacent} \}$ the \emph{neighborhood} of $v \in V(G)$, and by
$N[v] = N(v)\cup\{v\}$ the \emph{closed neighborhood} of $v$. 
We define  $u, v \in V(G)$ as \emph{true twins} (resp. \emph{false twins}) if $N[u] = N[v]$ (resp. $N(u) = N(v)$).
A vertex $v$ of $G$ is \emph{universal} if $N[v] = V(G)$.
We define the \emph{degree} of $v \in V(G)$, denoted by \emph{$d(v)$},  as the number of neighbors of $v$ in $G$, i.e. $d(v) = |N(v)|$

A \emph{clique} or \emph{complete set} (resp.\ \emph{stable set} or \emph{independent
set}) is a set of pairwise adjacent (resp.\ nonadjacent) vertices.
We use \emph{maximum} to mean maximum-sized, whereas
\emph{maximal} means inclusion-wise maximal. The use of
\emph{minimum} and \emph{minimal} is analogous. A vertex $v \in V(G)$ is said to be \emph{simplicial} if $G[N(v)]$ is a clique. 

A \emph{coloring} of a graph is an assignment of colors to its
vertices such that any two adjacent vertices are assigned
different colors. The smallest number $t$ such that $G$ admits a
coloring with $t$ colors (a \emph{$t$-coloring}) is called the
\emph{chromatic number} of $G$ and is denoted by $\chi(G)$. A
coloring defines a partition of the vertices of the graph into
stable sets, called \emph{color classes}.

A graph $G(V,E)$ is a \emph{comparability graph} if
there exists an ordering $v_1, \dots , v_n$ of $V$ such that, for
each triple $(r,s,t)$ with $r<s<t$, if $v_r v_s$ and $v_s v_t$ are
edges of $G$, then so is $v_r v_t$.\ Such an ordering is a
\emph{comparability ordering}. A graph is a \emph{co-comparability
graph} if its complement is a comparability graph.

A \emph{tree} $T$ is a connected graph that has no cycles. A \emph{rooted tree} $T_v$ is a tree in which a vertex $v \in V(T)$ is labeled as the \emph{root} of the tree. All vertices, known as \emph{nodes} of $T_v$, have implicit positions in relation to the root. 
Let $u,w \in V(T_v)$, $w$ is a \emph{descendant} of $u$ if the path from $w$ to $v$ includes $u$. The node $w$ is said to be a \emph{child} of $u$ if it is a descendant of $u$ and $(w,u) \in E(T_v)$. The \emph{children} of $u$ is defined as the set that containing all child nodes of $u$. The node $u$ is said to be a \emph{leaf} of $T_v$ if it has no child in $T_v$. 
The \emph{subtree} $T_u$ of $T_v$  \emph{rooted at} the node $u$ is the \emph{rooted} tree that consists of $u$ as the root  and its descendants in $T_v$ as the nodes.  

A \emph{directed graph}, or \emph{digraph}, is a graph $D = (V,E)$ such that $E$ consists of ordered pairs of $V(G)$. A \emph{directed cycle} of a digraph $D$ is a sequence $v_1,v_2, \ldots,v_i$, $1 \leq i \leq n$, of vertices of $V(D)$ such that $v_1 = v_i$ and, for all $1 \leq j<i$, $(v_j,v_{j+1}) \in E(D)$. A \emph{directed acyclic graph} (\emph{DAG}) $D$ is a digraph with no directed cycles.  A \emph{topological ordering} of a DAG $D$ is a sequence $v_1,v_2,\ldots ,v_n$ of $V(D)$ such that there are no $1\leq i <j\leq n$ such that $(v_j,v_i) \in E(D)$. Determining a topological ordering of a DAG can be done in time $O(n+m)$~\cite{Knuth1}.

An \emph{ordering} $s$ of elements of a set $\mathcal{C}$, denoted by $V(s)$, consists of a sequence $e_1,e_2,\ldots,e_n$  of all elements of $\mathcal{C}$. We define $\bar{s}$ as the \emph{reversal} of $s$, that is, $\bar{s} = e_n,e_{n-1},\ldots,e_1$. We say that $e_i$ \emph{precedes} $e_j$ in $s$, denoted by $e_i < e_j$, if $i < j$. An ordered tuple $(a_1, a_2, \ldots,  a_k)$ of some elements of $\mathcal{C}$ is \emph{ordered according to} $s$ when, for all $1 \leq i < k$, if $a_i = e_j$ and $a_{i+1} = e_z$, then $j < z$. Given orderings $s_1$ and $s_2$, the ordering obtained by concatenating $s_1$ to $s_2$ is denoted by $s_1s_2$.

\subsection{Interval graphs}
The \emph{intersection graph} of a family $\mathcal{F}$ of sets is the graph $G$ such that $V(G) = \mathcal{F}$ and $S,T \in V(G)$ are adjacent if and only if $S \cap T \neq \emptyset$.
An \emph{interval graph} $G$ is the intersection graph of a family $\mathcal{R}$ of intervals of the real line; such a family is called an \emph{interval model}, or  a \emph{model}, of the graph. 
We say that $\mathcal{R}$ is associated to $G$ and vice-versa.
It  is worth  mentioning that an interval graph can be associated to several models but an interval model can be associated to a unique graph. 
Concerning interval graphs, there are some characterizations which are relevant to this work, that we present next.

\begin{theorem}[\cite{Olariu91}]\label{theo:can:ordering}
A graph $G$ is an interval graph if, and only if, there is an ordering $s$ of $V(G)$ such that, for any triple $(p,q,r)$ of $V(G)$ ordered according to $s$, if $(p,r) \in E(G)$, then $(q,r) \in E(G)$.
\end{theorem}

The ordering described in Theorem~\ref{theo:can:ordering} is said to be a \emph{canonical} ordering.
Figure~\ref{fig:canonico} depicts an interval graph in which the vertices are presented from left to right in one of its canonical orderings. 

\begin{figure}
\centering     
\subfigure[]{\label{fig:canonico}\includegraphics[width=64mm]{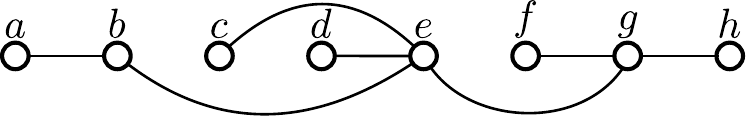}}
\hspace{0.45cm}
\subfigure[]{\label{fig:canonicoproprio}\includegraphics[width=64mm]{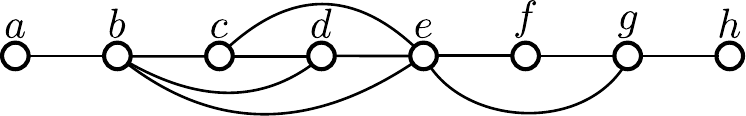}}
\caption{(a) Canonical ordering and (b) proper canonical ordering.}
\end{figure}

Let $\mathcal{R}$ be an interval model of an interval graph $G$. Note that, if we consider a vertical line that  intersects a subset of intervals in $\mathcal{R}$, then these intervals consist of a clique in $G$. This is true because they all contain the point in which this line is defined. Moreover, if the given line traverses a maximal set of intervals, then the corresponding clique is maximal. Figure~\ref{fig:model} depicts an interval model and the vertical lines that correspond to maximal cliques of the graph in Figure~\ref{fig:canonico}.
The following is a characterization of interval graphs in terms of the maximal cliques of the graph.

\begin{figure}[htb]
\centering     
\includegraphics[width=90mm]{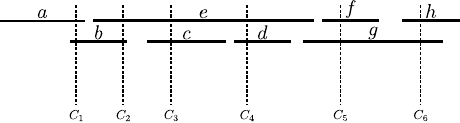}

\caption{ An interval model an the maximal cliques of the graph in Figure~\ref{fig:canonico}.}
\label{fig:model}
\end{figure}

\begin{theorem}[\cite{Roberts69}]\label{theo:clique:ordering}
A graph $G$ is an interval graph if, and only if, there is an ordering $s_C = C_1,C_2, \ldots, C_k$ of its maximal cliques such that if $v \in C_i \cap C_z$  with $1 \leq i \leq z \leq k$, then $v \in C_j$, for all $i \leq j \leq z$.
\end{theorem}

The ordering described in Theorem~\ref{theo:clique:ordering} is said to be a \emph{canonical clique} ordering. The ordering of the maximal cliques depicted in Figure~\ref{fig:model}, read from left to right, represents a canonical clique ordering.

A \emph{proper interval graph} is an interval graph that admits an interval model in which no interval properly contains another. There is a characterization of proper interval graphs which is 
similar to that presented in Theorem~\ref{theo:can:ordering}, as described next.

\begin{theorem}[\cite{Roberts69}]\label{theo:procan:ordering}
A graph $G$ is a proper interval graph if, and only if, there is an ordering $s$ of $V(G)$  such that, for any triple $(p,q,r)$ of $V(G)$ ordered according to $s$, if $(p,r) \in E(G)$, then $(p,q), (q,r) \in E(G)$.
\end{theorem}

The ordering specified in the previous theorem is defined as a \emph{proper canonical} ordering. The ordering of the vertices of the graph depicted in Figure~\ref{fig:canonicoproprio} from left to right is proper canonical. Note that a proper canonical ordering is also a canonical ordering.
An important property of proper canonical orderings is the following: 

\begin{lemma}[\cite{Roberts69}]
If $G$ is a connected proper interval graph. Then a proper canonical ordering of $G$ is unique up to reversion and permutation of mutual true twin vertices.
\end{lemma}

As a consequence, if a proper interval graph $G$ is disconnected, since each component has a unique proper canonical ordering, then the canonical ordering of $G$ consists of a permutation of canonical orderings of each of those components. 

Let $s$ be an ordering of $V(G)$ and $s_C = C_1,C_2, \ldots, C_k$ an ordering of the maximal cliques of $G$. 
The sequence $s$ is said to be \emph{ordered according} to $s_C$ if for all $u,v \in V(G)$ such that $u<v$ in $s$, there are no $ 1\leq i <j \leq k$ such that $v \in C_i \setminus C_j$ and  $u \in C_j$. As an example, note that the canonical ordering of the Figure~\ref{fig:canonico} is ordered according to  the maximal clique ordering $s_C = C_1,C_2, \ldots, C_6$ of the Figure~\ref{fig:model}.
The following lemma relates the characterization of canonical orderings and canonical clique orderings.

\begin{lemma}\label{lemma:clique:canonical}
Let $G$ be a graph and $s$ be an ordering of $V(G)$.  The ordering $s$ is a canonical ordering of $V(G)$ if, and only if, $s$ is ordered according to a canonical clique ordering of $G$.
\end{lemma}

\begin{proof}
Consider $s_C$ a canonical clique ordering of $G$. We will show that is possible to build a canonical ordering $s$ from $s_C$, respecting its clique ordering. To achieve this, first start with an empty sequence $s$. Iteratively, for each element $X$ of the sequence $s_C$, choose all simplicial vertices of $X$, adding them to $s$ in any order and removing them from $G$. 
Note that some of the vertices of $X$ that were not removed from $G$, because they were not simplicial in $G$, now may be turned simplicial by the removal of vertices.
Clearly, at the end of the process, $s$ will contain all the vertices of $G$.
Suppose $s$ is not a canonical ordering, that is, there are $p,q,r \in V(G)$, $p<q<r$ in $s$, such that $(p,r) \in E(G)$ and $(q,r) \not \in E(G)$. 
Let $C_p$, $C_q$ and $C_r$ be the maximal cliques of $s_C$ being processed at moment $p$, $q$ and $r$ were choose, respectively.  Note that $C_p < C_q < C_r$ in $S_c$ and, as $C_p$ is the last maximal clique in $s_C$ that contains $p$ and $(p,r) \in E(G)$,  $r \in C_p$. Besides, as $(q,r) \not \in E(G)$, $r \not\in C_q$ . Therefore, $r \in C_p \cap C_r$ and $r \not\in C_q$. A contradiction with the fact that $s_C$ is a canonical clique ordering. Hence, $s$ is a canonical ordering.

Consider $s= v_1,v_2, \ldots, v_n$ be a canonical ordering of $G$. We prove by induction in $|V(G)| = n$ that $s$ is ordered according to a canonical clique ordering. Clearly, the statement is true for $n = 1$. Suppose that the statement is true for any  $1 \leq n'< n$.
Let $s'$ be the sequence obtained from $s$ by removing $v_n$. Clearly $s'$ is also a canonical ordering and, by the induction hypothesis, $s'$ is ordered according with  a canonical clique ordering $s'_C = C_1,C_2, \ldots, C_k$. Let $C_i$, with $1 \leq i \leq k$, be the first clique of $s'_C$ such that $v_n \in N(C_i)$. 
Let $C' _j = \{C_j \cap N(v_n)\} \cup \{v_n\}$, $i \leq j \leq k$.
Note that, as $s$ is a canonical ordering, for all $C_j$ of $s'_C$, $i < j \leq k$, $\{C_j \cap N(v_n)\}\cup \{v_n\}$ is a maximal clique of $G$. 
If $C_i \subseteq N(v_n)$, then  $s'_C = C_1,C_2, \ldots, C'_i,C'_{i+1}, \ldots, C'_k$ is a canonical clique ordering that matches $s$. Otherwise, $s'_C = C_1,C_2, \ldots, C_i,C'_i,C'_{i+1}, \ldots, C'_k$ is a canonical clique ordering that matches $s$. Therefore, $s$ is ordered according to a canonical clique ordering of $G$.
\end{proof}

\subsubsection {PQ trees}

A \emph{$PQ$ tree}~\cite{BOOTH76} $T$ is a data structure consisting of an ordered tree that describes a family $\mathcal{F}$ of permutations of elements from a given set $\mathcal{C}$. 
In a $PQ$ tree $T$, the set of leaves is $\mathcal{C}$ and the permutation being represented by $T$ is the sequence of the leaves from left to right. Regarding  internal nodes, they are classified into two types, the $P$ and the $Q$ nodes. 
An \emph{equivalent} $PQ$ tree $T'$ to $T$ is a tree obtained from $T$ by any sequence of consecutive transformations, each consisting of either permuting the children of a $P$ node, or reversing the children of a $Q$ node. The family $\mathcal{F}$ of permutations of elements from $\mathcal{C}$ represented by  $T$ is that of permutations corresponding to all equivalent $PQ$ trees to $T$. 

Graphically, in a $PQ$ tree, leaves and $P$ nodes are represented by circles and $Q$ nodes by rectangles. In representations of schematic $PQ$ trees, a node represented by a circle over a rectangle will denote that, in any concrete $PQ$ tree conforming the scheme, such a node is either a $P$ node, or a $Q$ node.
Figure~\ref{fig:PQT:exeOp} depicts the described operations. In Figure~\ref{fig:PQT:1}, we have a partial representation of a $PQ$ tree. Figure~\ref{fig:PQT:2} depicts an equivalent $PQ$ tree obtained from a permutation of children of a $P$ node of this tree and Figure~\ref{fig:PQT:3} exemplifies an equivalent $PQ$ tree from the reversion of the children of a $Q$ node. 

    

\begin{figure}[htbp]
    \centering 
    \subfigure[]{\label{fig:PQT:1}\includegraphics[width=37.7mm]{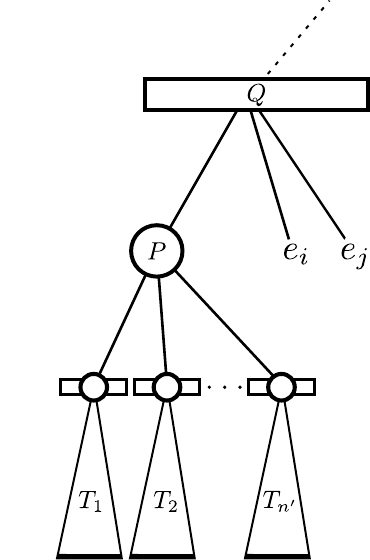}}
    \subfigure[]{\label{fig:PQT:2}\includegraphics[width=37.7mm]{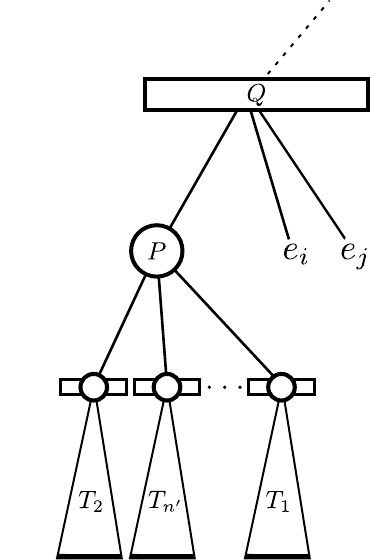}}\hfill
    \subfigure[]{\label{fig:PQT:3}\includegraphics[width=37.7mm]{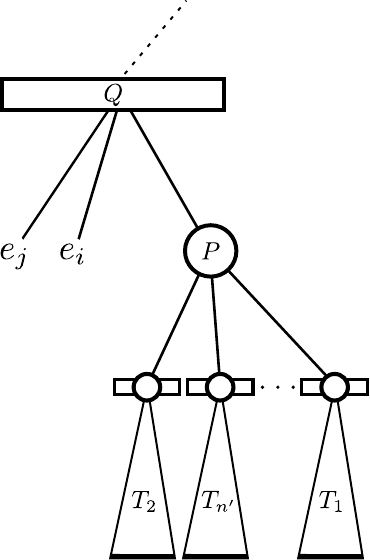}}\\\
    
\caption{A $PQ$ tree (a) and an example of permutations (b) and reversions (c) of children of its nodes. }
\label{fig:PQT:exeOp}
\end{figure}

One of the applications from the seminal paper introducing $PQ$ trees is that of recognizing interval graphs. 
In such an application, each  leaf of the $PQ$ tree  is a maximal clique of an interval graph $G$ and the family of permutations the tree represents is precisely all the canonical clique orderings of $G$~\cite{BOOTH76}. A $PQ$ tree can be constructed from an interval graph in time $O(n+m)$~\cite{BOOTH76}. 

We will say that a  vertex $v$ belongs to a node $X$ of a $PQ$ tree, and naturally denote by $v \in X$,  if it belongs to any leaf that descends from this node. Figure~\ref{fig:pre:PQT} depicts a $PQ$ tree of the interval graph of Figure~\ref{fig:canonico} according to the maximal cliques in  Figure~\ref{fig:model}. The permutations implicitly represented by this  $PQ$ tree are:

\begin{minipage}[t]{0.9\linewidth}
    \vspace{2ex}
    \begin{minipage}[t]{0.4\linewidth}
        \centering
        \begin{itemize}
            \item $C_1,C_2,C_3,C_4,C_5,C_6$
            \item $C_1,C_2,C_3,C_4,C_6,C_5$
            \item $C_1,C_2,C_4,C_3,C_5,C_6$
            \item $C_1,C_2,C_4,C_3,C_6,C_5$
            \item $C_6,C_5,C_4,C_3,C_2,C_1$
            \item $C_5,C_6,C_4,C_3,C_2,C_1$
            \item $C_6,C_5,C_3,C_4,C_2,C_1$
            \item $C_5,C_6,C_3,C_4,C_2,C_1$
        \end{itemize}
    \end{minipage}
    \hfill
    \begin{minipage}[t]{0.4\linewidth}
            \captionsetup{type=figure}
            \vspace{-1.5ex}
            \centering
            \includegraphics[width=40mm]{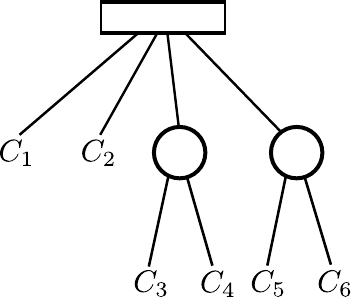}
            \caption{A $PQ$ tree of the interval graph of Figure~\ref{fig:canonico} according to the maximal cliques in  Figure~\ref{fig:model}.}
            \label{fig:pre:PQT}
    \end{minipage}
\end{minipage}

    


\subsection {Thinness and proper thinness}

A graph $G$ is called a \emph{$k$-thin graph} if there is a $k$-partition $(V_1,V_2 \ldots, V_k)$ of $V(G)$ and an ordering $s$ of $V(G)$ such that, for any triple $(p,q,r)$ of $V(G)$ ordered according to $s$, if $p$ and $q$ are in a same part $V_i$ and $(p,r) \in E(G)$, then $(q,r)\in E(G)$. An ordering and a partition satisfying that property are called \emph{consistent}. That is, a graph is  $k$-thin if there is an ordering consistent with some $k$-partition of its vertex set. 
The \emph{thinness} of $G$, denoted by \emph{$thin(G)$}, is the minimum  $k$ for which $G$ is a $k$-thin graph. 

A graph $G$ is called a \emph{proper $k$-thin graph} if $G$ admits a $k$-partition $(V_1, \ldots, V_k)$ of $V(G)$ and an ordering $s$ of $V(G)$ consistent with the partition and, additionally, for  any triple $(p,q,r)$ of $V(G)$, ordered according to $s$, if $q$ and $r$ are in a same part $V_i$ and $(p,r) \in E(G)$, then $(p,q)\in E(G)$. Equivalently, an ordering $s$ of $V(G)$ such that $s$ and its reverse are consistent with the partition. Such an ordering and partition are called \emph{strongly consistent}.
The \emph{proper thinness} of $G$, or \emph{$pthin(G)$}, is the minimum $k$ for which $G$ is a proper $k$-thin graph.

Figures~\ref{fig:c4thin} and~\ref{fig:c4pthin} depict two bipartitions of a graph, in which the classes are represented by distinct colors, and two different vertex orderings. The ordering of Figure~\ref{fig:c4thin} is consistent with the corresponding partition but not strongly consistent, while the ordering of Figure~\ref{fig:c4pthin} is strongly consistent with the corresponding partition.  

 Note that $k$-thin graphs (resp. proper $k$-thin graphs) generalize interval graphs (resp. proper interval graphs). The $1$-thin graphs (resp. proper $1$-thin graphs) are the interval graphs (resp. proper interval graphs). The parameter $\thin(G)$ (resp. $\pthin(G)$) is in a way a measure of how far a graph is from being an interval graph (resp. proper interval graph). 

For instance, consider the graph $C_4$. Since $C_4$ is not an interval graph, $\pthin(C_4) \geq \thin(C_4) > 1$. Figure~\ref{fig:c4} proves that $\thin(C_4) = \pthin(C_4) = 2$.

A characterization of $k$-thin or proper $k$-thin graphs by forbidden induced subgraphs is only known for $k$-thin graphs within the class of cographs~\cite{Bono18}. Graphs with arbitrary large thinness were presented in~\cite{Carl07}, while in~\cite{Bono18} a family of interval graphs with arbitrary large proper thinness was used to show that the gap between thinness and proper thinness can be arbitrarily large. The relation of thinness and other width parameters of graphs like boxicity, pathwidth, cutwidth and linear MIM-width was shown in~\cite{Carl07,Bono18}. 

Let $G$ be a graph and $s$ an ordering of its vertices. The graph $G_{s}$ has $V(G)$ as vertex set, and $E(G_{s})$ is such that for $v < w$, $(v,w) \in E(G_{s})$ if and only if there is a vertex $z$ in $G$ such that $v < w < z$, $(z,v) \in E(G)$ and $(z,w) \not \in E(G)$. 
Similarly, the graph $\tilde{G}_s$ has $V(G)$ as vertex set, and $E(\tilde{G}_{s})$ is such that for $v < w$, $(v,w) \in E(\tilde{G}_{s})$ if and only if either there is a vertex $z$ in $G$ such that $v < w < z$, $(z,v) \in E(G)$ and $(z,w) \not \in E(G)$ or there is a vertex $x$ in $G$ such that $x < v < w$, $(x,w) \in E(G)$ and $(x,v) \not \in E(G)$.

\begin{theorem}\label{thm:coloring-order}\cite{Bono18,B-M-O-thin-tcs} Given a graph $G$ and an ordering ${s}$ of its vertices, a partition of $V(G)$ is consistent (resp. strongly consistent) with the ordering ${s}$ if and only if the partition is a valid coloring of $G_{s}$ (resp. $\tilde{G}_s$), which means that each part corresponds to a color in the coloring under consideration. 
\end{theorem}



\begin{figure}[htbp]
    \centering 
\subfigure[]{\label{fig:c4thin}\includegraphics[width=40mm]{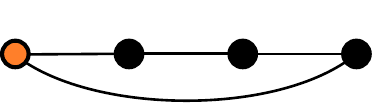}}\hfil
\subfigure[]{\label{fig:c4pthin}\includegraphics[width=40mm]{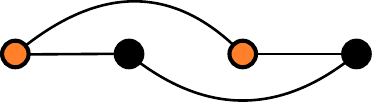}}\hfil

\caption{(a) A consistent ordering and a (b) strongly consistent ordering of $V(C_4)$, for the corresponding $2$-partitions.}
\label{fig:c4}
\end{figure}

\subsection {Precedence thinness and precedence proper thinness}

In this work, we consider a variation of the problems described in the last subsection by requiring that, given a vertex partition, the (strongly) consistent orderings hold an additional property. 

A graph $G$ is \emph{\fpExt{k}} (resp. \emph{\fppExt{k}}), or \emph{\fp{k}} (resp. \emph{\fpp{k}}), if there is a $k$-partition of its vertices and a  consistent (resp. strongly consistent) ordering $s$ for which the vertices that belong to a same part are consecutive in $s$. 
Such an ordering is called a 
\emph{precedence consistent ordering} (resp. \emph{precedence strongly consistent ordering}) for the given partition.  We define \fpPar{G} (resp. \fppPar{G}) as the minimum value $k$ for which  $G$ is \fp{k} (resp. \fpp{k}). 

The Figure~\ref{fig:C4-2ppt} illustrates a graph that is a \fpp{2} graph. The convention assumed is that the strongly consistent ordering being represented consists of the vertices ordered as they appear in the figure from bottom to top and, for vertices arranged in a same horizontal line, from left to right. Therefore, the strongly consistent ordering represented in Figure~\ref{fig:C4-2ppt} is $s =a,b,c,a',b',c'$.  The graph $C_4$ is not \fpp{2}, despite $\pthin(C_4)=2$. It can be easily verified by brute-force that, for all possible bipartitions of its vertex set and for all possible orderings $s$ in which the vertices of a same part are consecutive in $s$, the ordering and the partition are not strongly consistent. On the other hand, a \fpp{k} graph is a proper $k$-thin graph. Therefore, the class of \fpp{k} graphs is a proper subclass of that of proper $k$-thin graphs.

\begin{figure}[htbp]
    \centering 
\includegraphics[width=35mm]{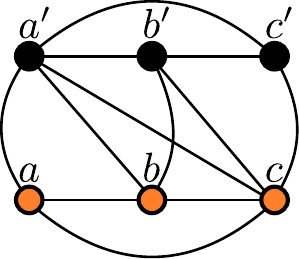}

\caption{A \fpp{2} graph.}
\label{fig:C4-2ppt}
\end{figure}

If a vertex order $s$ is given, by Theorem~\ref{thm:coloring-order}, any partition which is precedence (strongly) consistent with $s$ is a valid coloring of $G_s$ (resp. $\tilde{G}_s$) such that, additionally, the vertices on each color class are consecutive according to $s$. 
A greedy algorithm can be used to find a minimum vertex coloring with this property in polynomial time. Such method is described next, in Theorem~\ref{teo:kpptGreddy}.

\begin{theorem}\label{teo:kpptGreddy} 
Let $G$ be a graph and $s$ an ordering of $V(G)$. It is possible to obtain a minimum $k$-partition $\mathcal{V}$ of $V(G)$, in polynomial time, such that $s$ is a precedence (strongly) consistent ordering concerning  $\mathcal{V}$.
\end{theorem}

\begin{proof}

Consider the following greedy algorithm that obtains an optimum coloring of $G_s$ (resp. $\tilde{G}_s$) in which vertices having a same color are consecutive in $s$. That is, $s$ is a precedence consistent (resp. precedence strongly  consistent) ordering concerning the partition defined by the coloring.
Color $v_1$ with color~$1$. For each $v_i$, $i>1$,  let $c$ be the last color used. Then color $v_i$ with color $c$ if there is no $v_j$, $j<i$, colored with $c$ such that $(v_j,v_i) \in E(G_s)$ (resp. $(v_j,v_i) \in E(\tilde{G_s})$). Otherwise, color $v_i$ with color $c+1$. We show, by induction on $|s| = n$, that the algorithm finds an optimal coloring in which each vertex has the least possible color.

The case where $n=1$ is trivial. Suppose that the algorithm obtains an optimum coloring of $G_s$ (resp. $\tilde{G_s}$), for orderings having size less than $n$. 
Remove the last vertex $v_n$ from $s$ and $G_s$ (resp. $\tilde{G_s}$) and use the given algorithm to color the resulting graph. By the induction hypothesis, the chosen coloring for $v_1,v_2,\ldots,v_{n-1}$ is optimal. Moreover, the colors are non-decreasing and each vertex is colored with the least possible color. Now, add the removed vertex $v_n$ to $s$ and to the graph $G_s$ (resp. $\tilde{G_s}$), with its respective edges, and let the algorithm choose a coloring for it. If the color of $v_n$ is equal to the color of $v_{n-1}$ the algorithm is optimal by the induction hypothesis. Otherwise, $v_n$ is colored with a new color $c'$. Suppose the chosen coloring is not optimal. That is, it is possible  to color $v$ with an existing color.  This implies  that there is at least a neighbor $v_j$ of $v_n$, in $G_s$ (resp. $\tilde{G_s}$), that can be recolored with a smaller color. This is an absurd because the algorithm has already chosen the least possible color for all the vertices of $G_s \setminus \{v_n\}$ (resp. $\tilde{G_s}\setminus \{v_n\}$), relative to $s$. 
\end{proof}



In the following sections, we will deal with the case where the vertex partition is given and the problem consists of finding the vertex ordering. From now on, we will then simply call \emph{precedence consistent ordering} (resp. \emph{precedence strongly consistent ordering}) to one that is such for the given partition. 
\unskip
\section{Precedence thinness for a given partition}\label{chpt:pthinness}

In this section, we present an efficient algorithm to \fpExt{k} graph recognition for a given partition. This algorithm uses $PQ$ trees and some related properties to validate precedence consistent orderings in a greedy fashion, iteratively choosing an appropriate ordering of the parts of the given partition that satisfies precedence consistence, if one does exist.
Formally, the problem addressed in this chapter is the following. 

\begin{flushleft}
\begin{tabular}{ l l } 
 \hline
 \textbf{\textsc{Problem}:} & \textsc{Partitioned \fp{k}} (Recognition of \fp{k} graphs for a given \\ & partition)\\ 
 \textbf{\textsc{Input}:} & A natural $k$, a graph $G$ and a partition $(V_1,\ldots,V_k)$ of $V(G)$.\\ 
 \textbf{\textsc{Question}:}  & Is there a consistent ordering $s$ of $V(G)$ such that the vertices\\ & of $V_i$ are consecutive in  $s$, for all $1 \leq i \leq k$? \\ 
 \hline
\end{tabular}
\end{flushleft}

It should be noted that a precedence consistent ordering $s$ consists of  a concatenation of the consistent orderings of $G[V_i]$, for all $1 \leq i \leq k$.  That is, $s = s_1s_2\ldots s_k$, where $s_1,s_2,\ldots,s_k$ is a permutation of $s'_1,s'_2,\ldots,s'_k$ and $s'_i$ is a canonical ordering of $G[V_i]$, for all $1 \leq i \leq k$. The following property is straightforward from the definition of a precedence consistent ordering.

\begin{property}\label{prop:thinness}
Let  $(V_1,V_2, \ldots, V_k)$ be a partition of $V(G)$, $s$ a precedence consistent ordering and  $1 \leq i , j \leq k$. If $V_i$ precedes $V_j$ in $s$, then, for all  $u,v \in V_i$ and $w \in V_j$, if $(u,w) \not\in E(G)$ and $(v,w) \in E(G)$, then  $u$ precedes $v$ in $s$. 
\end{property}

Property~\ref{prop:thinness} shows that, for any given consistent ordering $s = s_1s_2\ldots s_k$, the vertices of $s_j$ impose ordering restrictions on the vertices of $s_i$, for all $1 \leq i < j \leq k$. This relation is depicted in Figure~\ref{fig:thinness}. This property will be used as a key part of the greedy algorithm to be presented later on.

 \begin{figure}[htbp]
\centering     
\includegraphics[width=80mm]{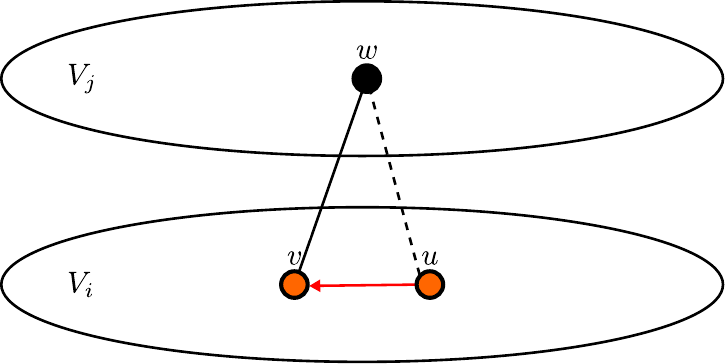}
\caption{Precedence relations among the vertices in a precedence consistent ordering.}
\label{fig:thinness}
\end{figure}

Let  $s_T= C_1,C_2, \ldots, C_q$ be a ordering of the maximal cliques of an interval graph $G$ obtained from a $PQ$ tree $T$. Recall from Section~\ref{chpt:preliminaries} that it is possible to obtain a canonical ordering $s$ ordered according to $s_T$.
Let $u,v\in V(G)$. We define $T$ as \emph{compatible with the ordering restriction} $u<v$ if there exists a canonical ordering $s$ ordered according to  $s_T$ such that $u<v$ in $s$. 
The following theorem describes 
compatibility conditions between a $PQ$ tree and an ordering restriction $u<v$.

\begin{theorem}\label{theo:thinness}
Let $G$ be an interval graph  and $T$ be a $PQ$ tree of $G$. Let $X$ be a node of $T$ with children $X_1, \ldots, X_k$ and $u,v \in X$. Denote by $T_X$ the subtree rooted at $X$. The following statements are true.

\begin{enumerate}[(i)] 

\item if $v$ belongs to all leaves of $T_X$, then $T_X$ is compatible with $u<v$ (see Figure~\ref{fig:theo:thinness:1}).

\item Let  $X_i,X_j \in \{X_1,\ldots,X_k\}$. If $u \in X_i$,  $v\not \in X_i$, $v \in X_j$ and $T_X$ is compatible with $u<v$, then $X_i$ precedes $X_j$ in $T_X$ (see Figure~\ref{fig:theo:thinness:2}). 

\end{enumerate}

\end{theorem}

\begin{proof}\leavevmode \\
Let $G_X$ be the graph induced by the union of the leaves of $T_X$. 
\begin{enumerate}[(i)] 
\item Let $s_x$ be a canonical ordering of $G_X$ and $s'_x$ the ordering obtained from $s_x$ by moving $v$ to the last position. As $v$ is an universal vertex from $G_X$, $s'_x$ is also a canonical ordering of this graph. Consequently, $T_X$ is compatible with  $u<v$.

\item  Suppose $X_j$ precedes $X_i$ in $T_X$ and $T_X$ is compatible with $u<v$. As $v \not \in X_i$, then by Theorem~\ref{theo:clique:ordering}, there is no $X_z$ such that $X_i$ precedes $X_z$  and $v  \in X_z$. Otherwise, there would exist three maximal cliques $C_i,C_j,C_z$ such that $C_j<C_i<C_z$ in the ordering of cliques represented in $T_X$ and such that $v \not \in C_i$, $v \in C_j \cap C_z$. Therefore $v<u$ in any canonical ordering $s_X$ of $G$ (Lemma~\ref{lemma:clique:canonical}),  a contradiction because $T_X$ is compatible with $u<v$. Thus, $X_i$ precedes $X_j$ in $T_X$. \qedhere
\end{enumerate} 
\end{proof}

\begin{figure}[htbp]
\centering     
    \subfigure[]{\label{fig:theo:thinness:1}\includegraphics[width=43mm]{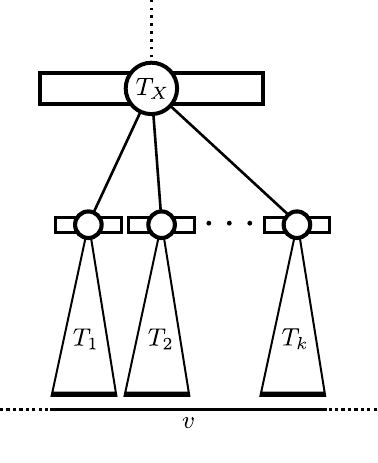}}
    \hfil
    \subfigure[]{\label{fig:theo:thinness:2}\includegraphics[width=43mm]{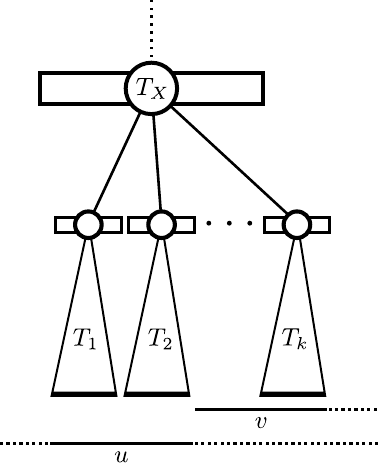}}
\caption{Ordering imposed by Theorem~\ref{theo:thinness} item $i$ ($a$) and item $ii$ ($b$).}
\end{figure}


Theorem~\ref{theo:thinness} can be used to determine the existence of a $PQ$ tree compatible with an ordering restriction $u<v$. This task can be achieved by considering as the node $X$ of Theorem~\ref{theo:thinness}, each one of the nodes of a given $PQ$ tree $T$. If $T$ violates the conditions imposed by the theorem, $T$ is ``annotated'' in a way that the set of equivalent $PQ$ trees is restricted, avoiding precisely the violations. This procedure continues until all nodes produce their respective restrictions, in which case any equivalent $PQ$ tree allowed by the ``annotated''  tree $T$ is compatible with the given ordering restrictions. If there is no equivalent $PQ$ tree to the ``annotated'' tree $T$, which means there is no way to avoid the violations, then there is no tree which is compatible with such an ordering restriction.

The general idea to ``anotate'' a $PQ$ tree is to use an auxiliary digraph to represent required precedence relations among the vertices. This digraph is constructed for each part and any topological ordering of it results in a precedence consistent ordering concerning the vertices of this part. Property~\ref{prop:thinness} is applied to determine the ordering restrictions of the vertices and Theorem~\ref{theo:thinness} to ensure that at the end of the algorithm the vertices are ordered according to a canonical clique ordering.  
Such a procedure is detailed next. 

First, the algorithm validates if each part of the partition induces an interval graph. This step can be  accomplished, for each part, in linear time~\cite{BOOTH76}. If at least one of these parts does not induce an interval graph, then the answer is \textsc{NO}. Otherwise, the algorithm tries each part as the first of a precedence consistent ordering. For each candidate part $V_i$, $1 \leq i \leq k$, it builds a digraph $D$ to represent the order conditions that the vertices of $V_i$  must satisfy in the case in which $V_i$ precedes all the other parts of the partition. That is, $V_i$ must be ordered in such a way that it is  according to a canonical clique ordering and respects the restrictions imposed by  Property~\ref{prop:thinness}. In this strategy, the vertex set of $D$ is $V_i$ and its directed edges represent the precedence relations among its vertices. Namely, $(u,v) \in E(D)$ if, and only if, $u$ must precede $v$ in all precedence consistent orderings that have $V_i$ as its first part.  
The algorithm  uses Property~\ref{prop:thinness} to find all the ordering restrictions $u < v$ among the vertices of $V_i$ imposed by the others parts, adding the related directed edges to $D$. 
Then, by building a $PQ$ tree $T$ of $G[V_i]$, Theorem~\ref{theo:thinness} is used to transform $T$ into $T'$ ensuring that all those ordering restrictions $u < v$ are satisfied. If there is a $PQ$ tree $T'$ of $G[V_i]$ that is compatible with all the ordering restrictions imposed by Property~\ref{prop:thinness}, then the algorithm adds directed edges to $D$ according to the canonical clique ordering represented by $T'$. This step is described below in the Algorithm~\ref{alg:pqtco} and it is similar to the one described in Lemma~\ref{lemma:clique:canonical}. At this point, $D$ is finally constructed and the existence of a topological ordering for its vertices determines whether $V_i$ can be chosen as the first part of a precedence consistent ordering for the given partition. If that is the case, $V_i$ is chosen as the first part and the process is repeated in $G \setminus V_i$ to choose the next part. If no part can be chosen at any step, the answer is \textsc{NO}. Otherwise, a feasible ordering of the parts and of the vertices within each part is obtained and the answer is \textsc{YES}. Next, the validation of the compatibility of $T$ is described in more detail.

\begin{algorithm}[htbp]
\textbf{Input}: $G$: an interval graph; 
        $D$: a digraph;
        $T$: a $PQ$-tree;\\
\vspace{0.20cm}
\addEdgesFromPQTree{
    Let $s_C$ be  the canonical clique ordering relative to $T$
    
    \For{\Each  $C_i \in s_C$ }{
        Let $S$ be  the set of simplicial vertices of $C_i$
        
        \For{\Each  $v \in S$ }{
            \For{\Each  $u \in C_{i+1}$}{
                $E(D) \leftarrow  E(D) \cup \{(v,u)\}$
            }
        }
        $G \leftarrow G \setminus S$
    }
}

\caption{\textsc{Adding  edges from $PQ$-tree to $D$}}
\label{alg:pqtco}
\end{algorithm}

For each imposed ordering $u<v$ (that is, $(u, v) \in E(D)$), $T$ is traversed  node by node, applying Theorem~\ref{theo:thinness}. As a consequence of such an application, if the order of the children of some node $X$ of $T$ must be changed to meet some restrictions, directed edges are inserted on $T$ to represent such needed reorderings. Those directed edges appear among nodes that are children of a same node in $T$. At the end, to validate if $T$ is compatible with all needed reorderings of children of nodes, a topological ordering is applied to the children of each node.
If there are no cycles among the children of each node, then there is an equivalent $PQ$ tree $T'$ compatible with all the restrictions. In this case, $T'$ is obtained from $T$ applying the sequence of permutations of children of $P$ nodes and reversals of children of $Q$ nodes which are compliant to the topological orderings.
Otherwise, if it is detected a cycle in some topological sorting, then there is no $PQ$ tree of $G[V_i]$ which is compatible  with all the set of restrictions. In other words, $V_i$ can not be chosen as the first part in a precedence consistent ordering for the given partition. 
Algorithm~\ref{alg:kfp} formalizes the procedure.

\begin{algorithm}[htbp]
\textbf{Input}: $G$: a graph; 
        $k$: a natural number; 
        $\mathcal{V}$: a $k$-partition $(V_1, V_2, \ldots, V_k)$ of $V(G)$;  \\
\vspace{0.20cm}
\Kpt{
    $s \leftarrow \emptyset$\\
    \For{\Each  $V_i \in \mathcal{V}$ }{
         \If{$G[V_i]$ \textsf{\upshape  is not an interval graph}}{
            \Return (\textsc{NO}, $\emptyset$)
         }
    }
    \While{$\mathcal{V} \neq \emptyset$}{
        \For{\Each $V_i \in \mathcal{V}$ }{
            $foundFirstPart$ $\leftarrow$ \TRUE \\
            Create a digraph $D = (V_i, \emptyset)$\\
            Build a $PQ$ tree $T_i$ of $G[V_i]$\\
            \For{\Each $V_j \in \mathcal{V}$ \textsf{\upshape such that } $V_j \neq V_i$}{
                Let $S$ be the set of precedence relations among the vertices of $V_i$ concerning $V_j$ (Property~\ref{prop:thinness})\\
                \For{\Each $(u<v) \in S$}{
                   $E(D) \leftarrow  E(D) \cup \{(u,v)\}$ \\
                    \For{\Each \textsf{\upshape node $X$  of $T_i$}}{
                        Add the direct edges, deriving from $(u<v)$, among the children of $X$ (Theorem~\ref{theo:thinness}) \\
                    }
                }
                \For{\Each \textsf{\upshape node} $X$ \textsf{\upshape of} $T_i$}{
                        Let $D_X = (V', E')$ be the digraph where $V'$ is the set of the children of $X$ and  $E'$ are the directed edges added among them\\
                        \If{\textsf{\upshape there is a topological ordering $s_X$ of $D_X$}}{
                             Arrange the children of $X$ according to $s_X$ \\
                        }
                        \Else{
                            $foundFirstPart \leftarrow$ \FALSE \\
                        }
                    }
            }
            \If{ foundFirstPart}{
                \AddEdgesFromPQTree{$G[V_i]$, $D$, $T_i$}\\
                \If{\textsf{\upshape there is a topological ordering $s_i$ of $D$}}{
                    $s \leftarrow ss_i$\\
                    $\mathcal{V} \leftarrow \mathcal{V} \setminus V_i$\\
                    \Break
                }
                \Else{$foundFirstPart \leftarrow$ \FALSE \\}
            }
        }
        \If{ \Not $foundFirstPart$}{
            \Return (\textsc{NO}, $\emptyset$)
        }
    }
    \Return (\textsc{YES}, $s$)
}
\caption{\textsc{Partitioned \fp{k}}}
\label{alg:kfp}
\end{algorithm}


To illustrate the execution of Algorithm~\ref{alg:kfp}, consider the graph $G$ as defined in Figure~\ref{fig:thinEx} and the $3$-partition $\mathcal{V} = (V_1,V_2,V_3)$ of $V(G)$ where $V_1 = \{a,b,c,d,e,f,g,h,i,j,k,l\}$, $V_2 = \{a',b',c',d',e',f',g,h',i',j',k',l'\}$ e $V_3 = \{a'',$ $b'',c'',d'',e'',f'',g'',h'',i'',j'',k'',l''\}$. For the sake of clearness, in Figure~\ref{fig:thinEx}, the edges with endpoints in distinct parts are depicted in black, the edges with endpoints in a same part are in light gray and the vertices belonging to distinct parts are represented with different colors. Moreover, the vertices of each part, read from the left to right, consist of a canonical ordering of the graph induced by that part.
Each  part of $\mathcal{V}$ induces the interval graph $G'$ depicted in Figure~\ref{fig:g1:graph}. Figure~\ref{fig:g1:model} depicts a model of $G'$. In this model, all maximal cliques are represented by vertical lines. Figure~\ref{fig:g1:PQT} represents a $PQ$ tree of $G'$ in which each maximal clique is labeled  according to the model in Figure~\ref{fig:g1:model}.

\begin{figure}[htbp]
    \centering 
    \includegraphics[width=90mm]{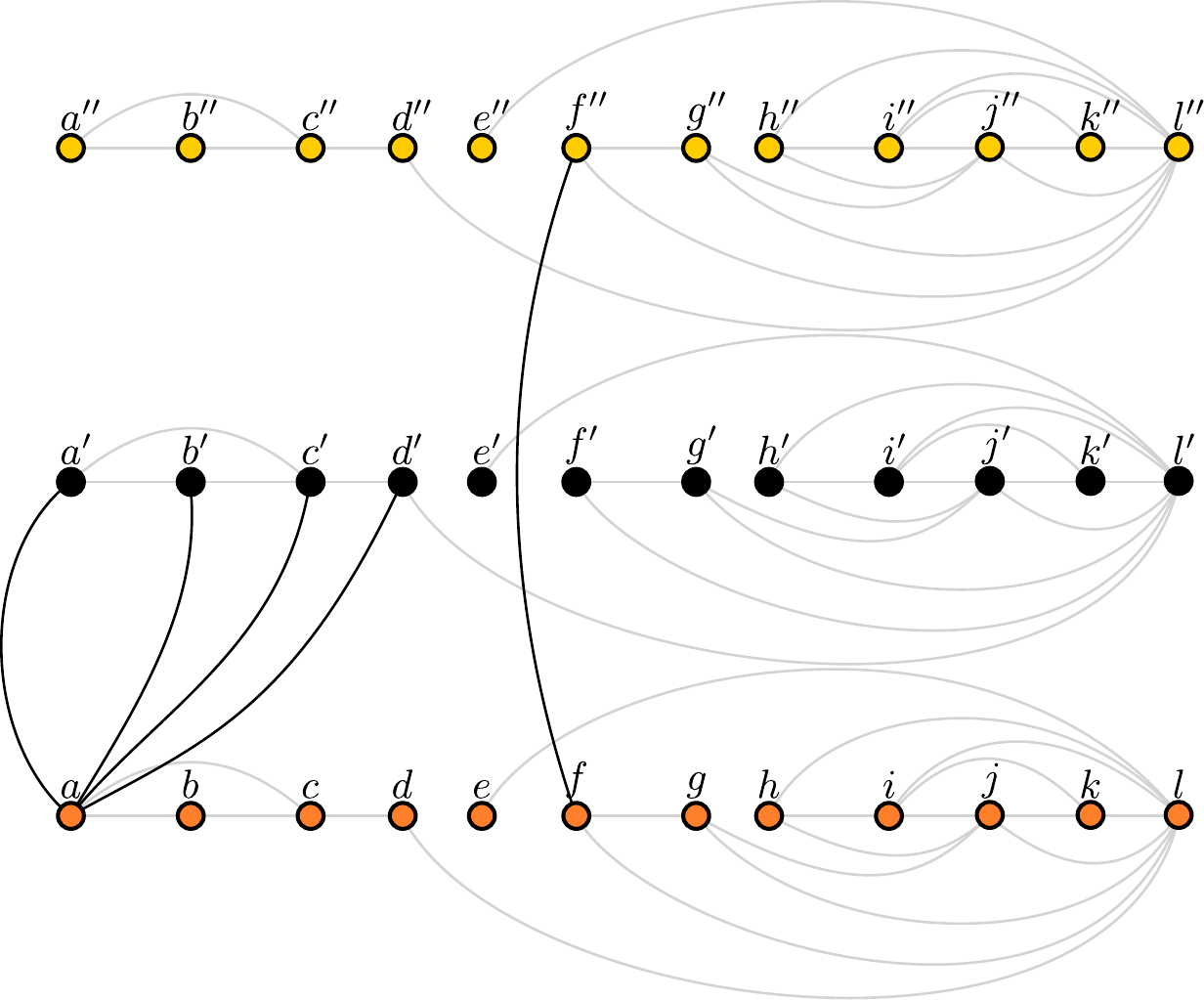}
\caption{A graph $G$ and a $3$-partition $\mathcal{V} = (V_1,V_2,V_3)$ of its vertices where  $V_1 = \{a,b,c,d,e,f,g,h,i,j,k,l\}$, $V_2 = \{a',b',c',d',e',f',g,h',i',j',k',l'\}$ e $V_3 = \{a'',$ $b'',c'',d'',e'',f'',g'',h'',i'',j'',k'',l''\}$.}
\label{fig:thinEx}
\end{figure}



\begin{figure}[htb]
    \centering 
    \subfigure[]{\label{fig:g1:graph}\includegraphics[width=80mm]{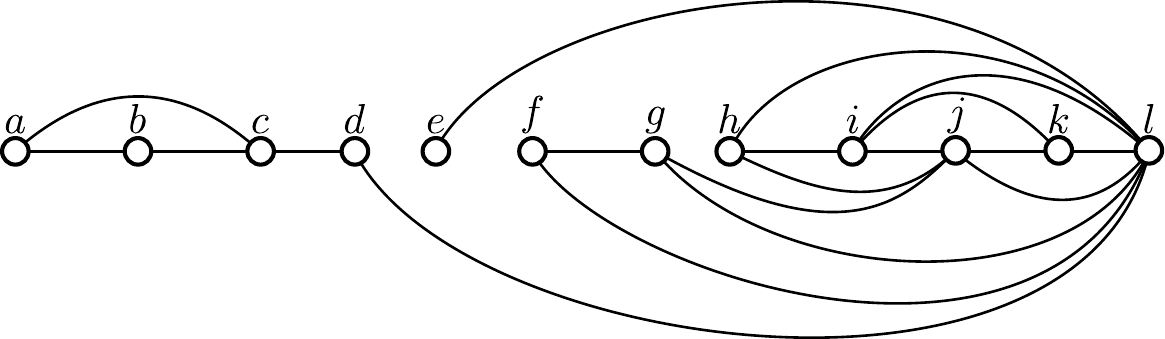}}\\
    \subfigure[]{\label{fig:g1:model}\includegraphics[width=90mm]{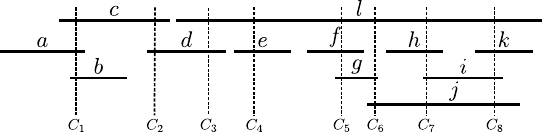}}\\
    \subfigure[]{\label{fig:g1:PQT}\includegraphics[width=44mm]{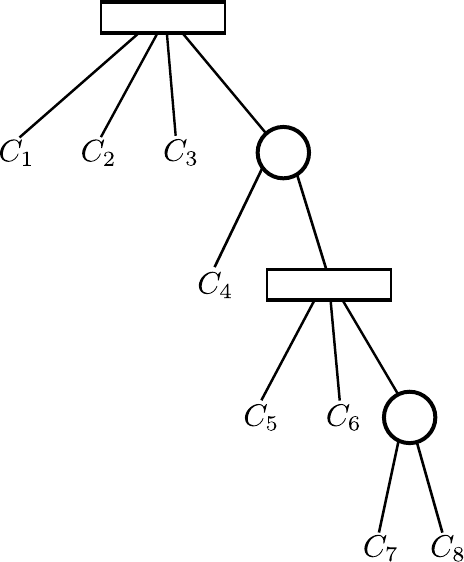}}
\caption{An interval graph $G'$ ($a$); an interval model ($b$)  and a $PQ$ tree (c) of $G$.}

\end{figure}

Suppose that, at the first step, the algorithm tries to choose $V_1$ as the first part of the precedence consistent ordering. As mentioned,  $G[V_1]\cong G'$ and, according to the model in Figure~\ref{fig:g1:model}, $G[V_1]$ has maximal cliques $C_1 = \{a,b,c\}$, $C_2=\{c,d\}$, $C_3 = \{d,l\}$, $C_4 = \{e,l\}$, $C_5 = \{f,g,l\}$, $C_6 = \{g,j,l\}$, $C_7 = \{h,i,j,l\}$ and $C_8 = \{i,j,k,l\}$. Concerning the edges between $V_1$ and $V_2$ and according to Property~\ref{prop:thinness}, the vertex $a$ of $V_1$ must succeed all the other vertices of this part in any valid canonical ordering. This requirement is translated into the corresponding $PQ$ tree through directed edges as depicted in Figure~\ref{fig:thinEX:PQT:v1v2}. Let $X$ be the current node of $T$. Note that if $X$ is a $Q$ node, then any imposed ordering of a pair of its children implies in  ordering all of them. In Figure~\ref{fig:thinEX:PQT:v1v2}, the oriented edges deriving from  Property~\ref{prop:thinness}  are represented in blue and the edges deriving from the orientation demanded by $Q$ nodes, due to the presence of the blue ones, are represented in orange. Clearly, there is a valid $PQ$ tree that satisfies such orientations. 
Now the algorithm adds the directed edges to the tree deriving from fact that  $V_1$ must also precede $V_3$. Considering the edges between $V_1$ and $V_3$, and according to Property~\ref{prop:thinness}, the vertex $f$ of $V_1$ must succeed all the other vertices of $V_1$ in any valid canonical ordering. This requirement is translated into the $PQ$ tree in Figure~\ref{fig:thinEX:PQT:v1v2}, resulting in  the $PQ$ tree in Figure~\ref{fig:thinEX:PQT:v1v3v1v2}. Clearly, no  $PQ$ tree can  satisfy the given orientation due to the directed cycle at the first level of the tree. Then, $V_1$ can not precede both $V_2$ and $V_3$.



\begin{figure}[htb]
    \centering 
\subfigure[]{\label{fig:thinEX:PQT:v1v2}\includegraphics[width=44mm]{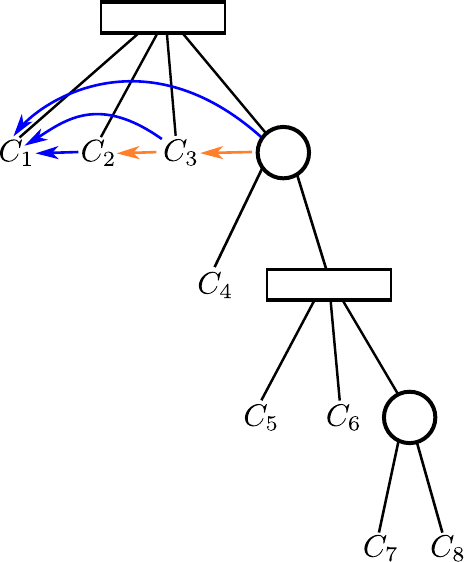}}\hfill
\subfigure[]{\label{fig:thinEX:PQT:v1v3v1v2}\includegraphics[width=44mm]{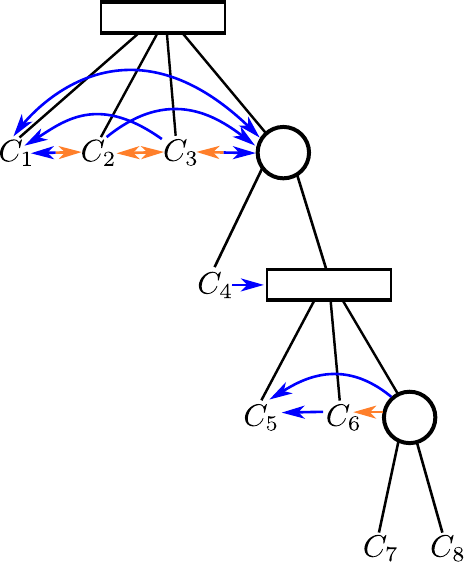}} \\
\subfigure[]{\label{fig:thinEX:PQT:v2v1}\includegraphics[width=44mm]{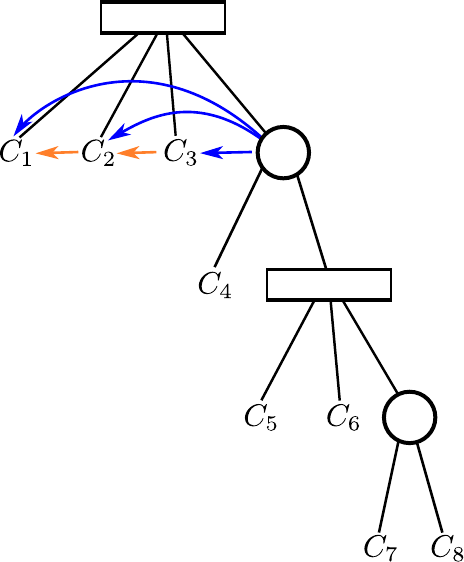}}\hfill  \subfigure[]{\label{fig:thinEX:PQT:v1v3}\includegraphics[width=44mm]{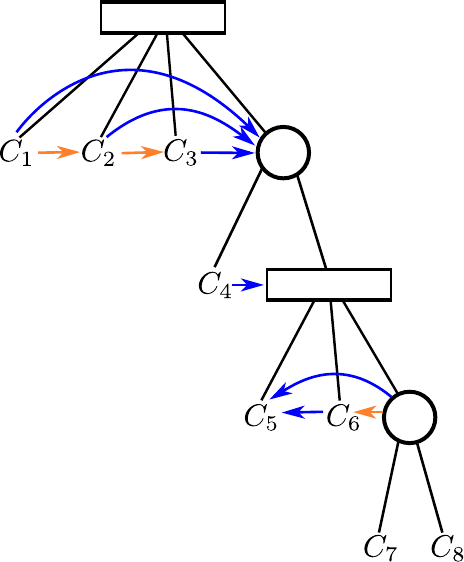}}

\caption{The edges added to the $PQ$ tree of Figure~\ref{fig:g1:PQT} through the example of execution of the Algorithm~\ref{alg:kfp}.}

\end{figure}

\begin{figure}[htb]
    \centering 
    \subfigure[]{\label{fig:v2:d}\includegraphics[width=90mm]{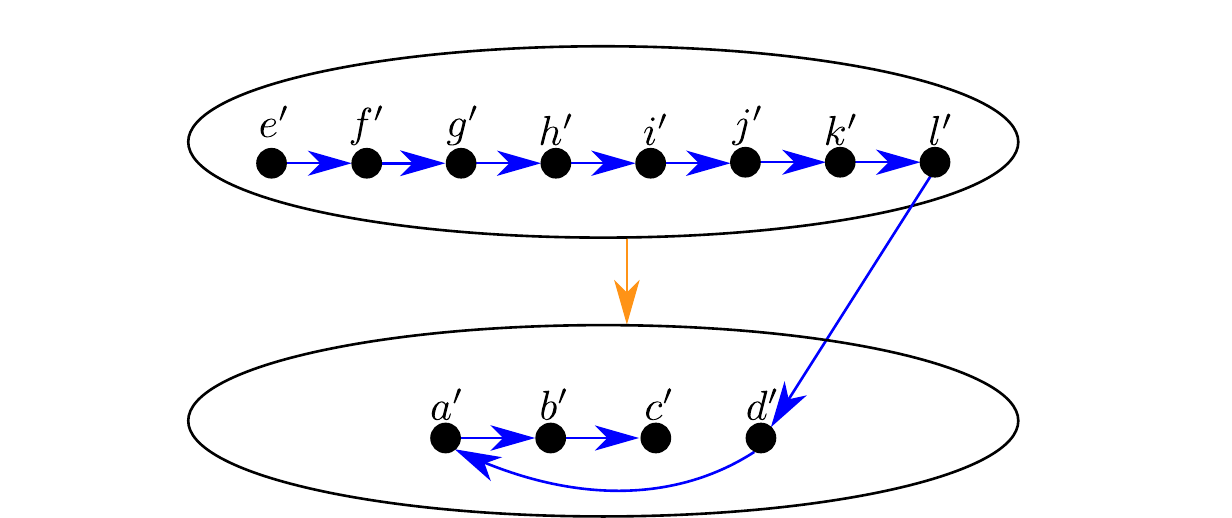}}
    \subfigure[]{\label{fig:v1:d}\includegraphics[width=90mm]{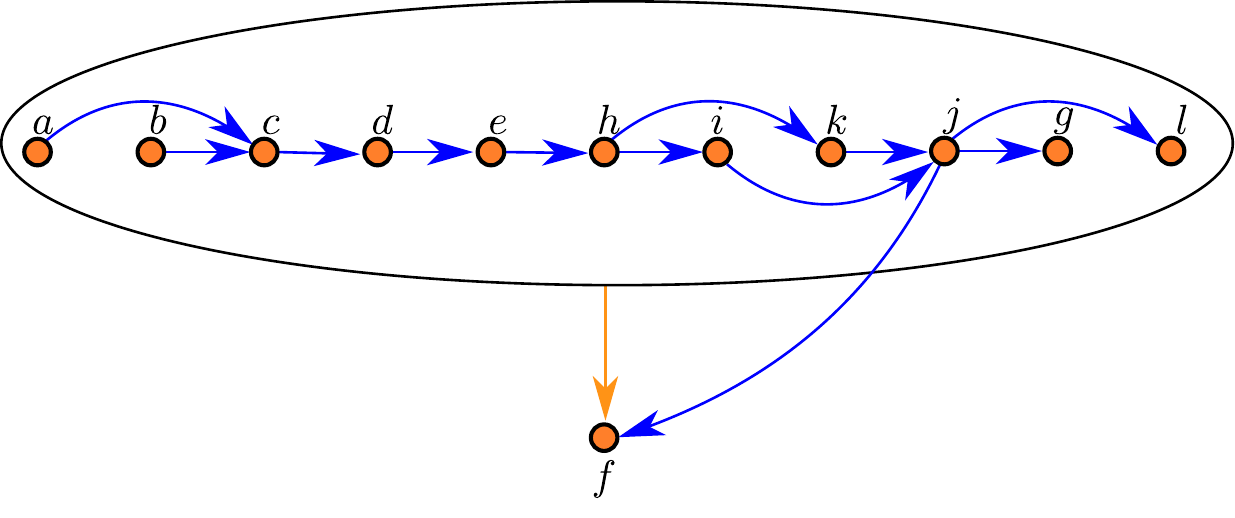}}\\
\caption{Related digraph $D$ when $V_2$ is chosen as the first part ($a$) and when $V_1$ is chosen as the second part ($b$) in Algorithm~\ref{alg:kfp}.}

\end{figure}
As $V_1$ can not be chosen as the first part, the algorithm tries another part as the first of the precedence consistent ordering. Suppose it now chooses $V_2$ as the first part. The interval graph $G[V_2]$ has as maximal cliques $C_1 = \{a',b',c'\}$, $C_2=\{c',d'\}$, $C_3 = \{d',l'\}$, $C_4 = \{e',l'\}$, $C_5 = \{f',g',l'\}$, $C_6 = \{g',j',l'\}$, $C_7 = \{h',i',j',l'\}$ and $C_8 = \{i',j',k',l'\}$. Note that, as there are no edges between $V_2$ and $V_3$, $V_2$ can precede $V_3$ in any valid consistent ordering. The algorithm must decide whether $V_2$ can precede $V_1$. According to the Figure~\ref{fig:thinEx} and  Property~\ref{prop:thinness}, the vertices $\{a',b',c',d'\}$ of $V_2$ must succeed all the other vertices of the same part  in any valid consistent ordering.
This requirement is again translated into directed edges in the $PQ$ tree of $V_2$  as depicted in Figure~\ref{fig:thinEX:PQT:v2v1}. Clearly, there is a $PQ$ tree $T'$ satisfying those directed edges. Then, the algorithm adds the edges deriving from the canonical clique ordering represented by $T'$ to $D$. Figure~\ref{fig:v2:d} depicts the final state of $D$ once the necessary edges has been added. In this figure, the edges deriving from Property~\ref{prop:thinness} are presented with an orange color and the edges related to $T'$ are presented in a blue color. For readability, in these figures the edges that can be obtained by transitivity are omitted. As there no cycles in $D$, $V_2$ can precede  $V_1$ and $V_3$. A topological ordering of $D$ leads to the canonical ordering $s =e',f',g',h',i',j',k',l',d',a',b',c'$ of $G[V_2]$.

After deciding $V_2$ as the first part, the algorithm uses the same process to choose the second part. Suppose it tries $V_1$ as the second part. Figure~\ref{fig:thinEX:PQT:v1v3} depicts the edges added to the $PQ$ tree of $G[V_1]$. As there are no cycles in the edges added, there is a tree which is compatible with the precedence relations associated with $V_1<V_3$. 
Figure~\ref{fig:v1:d} depicts the final state of the digraph $D$ related to $V_1$. As there are no cycles in $D$, $V_1$ can precede $V_3$, so the algorithm chooses it as the second part.  
Finally, the algorithm chooses $V_3$ as the last part and determine that there is a precedence consistent ordering such that $V_2<V_1<V_3$.

Concerning the complexity of the given strategy, each time  one of the $k$ parts is tried to be the first, we build a new digraph $D$, a new  $PQ$ tree $T$ 
and obtain the precedence relations according to Property~\ref{prop:thinness}. Enumerating all the precedence relations requires at most $O(n^3)$ steps, which is the time that it takes to  iterate over all triples of vertices of the graph. Moreover, each one of these relations  must be mapped to $D$, which takes time $O(1)$, and $T$. First, note that the number of nodes of $T$ is asymptotically bounded by its number of leaves,  that is, by the number of maximal cliques of the part being processed. As the number of maximal cliques is bounded by the number of vertices of the given part, the  number of nodes of $T$ is $O(n)$. Consequently, it is possible to model a precedence relation  of type $u<v$ into $T$, using Theorem~\ref{theo:thinness}, in time $O(n^2)$. To achieve this, first $T$ is traversed, in order to decide which nodes contain (resp. not contain) $u$ and $v$. A traversal of $T$ can be done in time $O(n)$, and $T$ can be constructed in $O(n+m)$ time. Additional steps will be necessary and generate new traversals in $T$ following the tree levels, with the purpose to add the necessary directed edges among the vertices that are children of the same node. This step can be done in $$\sum_{v \in V(T)}{ d^2(v)} \leq  \sum_{v \in V(T)}{d(v)|V(T)|} = O(|E(T)| |V(T)|) = O(|V(T)|^2) = O(n^2)$$
\noindent as $|E(T)| = O(|V(T)|)$ and $|V(T)| = O(n)$. 
Aiming to verify the existence of a compatible tree, the algorithm applies a topological ordering to the children of each node of $T$, which takes overall time $O(n(n+m))$. Then, the ordering in $T$ is translated to $D$ through directed edges. By using the Algorithm~\ref{alg:pqtco}, this step requires no more than $O(n^2+m)$ operations. 
Finally, a topological ordering is applied to $D$.
Thus, the algorithm has $$O(k^2(n+m+n^3n^2+n^2+nm+n^2+m+n+m))  = O(k^2n^5)$$ time complexity.

    





\unskip
\section{Precedence Proper Thinness for a Given Partition}\label{chpt:ppthinness}

In this section, we discuss precedence proper thinness for a given partition. First, we prove that this problem is \NP-complete for an arbitrary number of parts. Then, we propose a polynomial time algorithm for a fixed number of parts based on the one presented in Section~\ref{chpt:pthinness}.
Formally, we will prove that the following problem is \NP-complete.

\begin{flushleft}
\begin{tabular}{ l l } 
 \hline
 \textbf{\textsc{Problem}:} & \textsc{Partitioned \fpp{k}} (Recognition of \fpp{k}  graphs for a given\\ & partition)\\
 \textbf{\textsc{Input}:} & A natural $k$, a graph $G$ and a partition $(V_1,\ldots,V_k)$ of $V(G)$.\\ 
 \textbf{\textsc{Question}:} & Is there a strongly consistent ordering $s$ of $V(G)$ such that  the\\& vertices of $V_i$ are consecutive in  $s$, for all $1 \leq i \leq k$?\\ 
 \hline
\end{tabular}
\end{flushleft}

The \NP-hardness of the previous problem is accomplished by a reduction from the problem  \textsc{Not all equal 3-SAT}, which is \NP-complete~\cite{Schaefer1978b}. The details are described in Theorem~\ref{theo:npc:ppthinness}.  

\begin{flushleft}
    \begin{tabular}{ l l } 
        \hline
        \textbf{\textsc{Problem}:} & \textsc{Not all equal 3-SAT}\\ 
        \textbf{\textsc{Input}:} & A formula $\varphi$ on variables $x_1,\dots, x_r$ in conjunctive normal form,\\&  with clauses  $\mC_1, \dots, \mC_s$, where each clause has exactly three lite-\\&rals.\\
        \textbf{\textsc{Question}:} & Is there a truth assignment for $x_1,\dots, x_r$ such that each clause \\&$\mC_i$,  $i = 1, \dots, s$, has at least one true literal and at least one fal-\\&se literal?
        \\ 
        \hline
    \end{tabular}
\end{flushleft}

\begin{theorem}\label{theo:npc:ppthinness}
Recognition of \fpp{k} graphs for a given partition is
\NP-complete, even if the size of
each part is at most 2. 
\end{theorem}

\begin{proof} 
A given precedence strongly consistent ordering for the partition of $V(G)$ can be easily verified in polynomial time. Therefore, this problem is in \NP.

Given an instance $\varphi$ of \textsc{Not all equal 3-SAT}, we
define a graph $G$ and a partition of $V(G)$ in which each part has size at most two. The graph $G$ is defined in such a way that $\varphi$ is satisfiable if, and only if, there is a precedence strongly consistent ordering of  $V(G)$ for the partition. The graph $G$ is constructed as follows.

For each variable $x_i$ appearing in the clause $\mC_j$, create
the  part $$X_{ij} = \{x_{ij}^T,x_{ij}^F\}$$

For each variable $x_i$, create the  parts $$X_{i}^T =
\{x_{i}^T\} \mbox{ and } X_{i}^F = \{x_{i}^F\}$$

The edges of the graph between these parts are $(x_i^T, x_{ij}^T)$
and $(x_i^F, x_{ij}^F)$ for every $i,j$ such that variable $x_i$
appears in clause $\mC_j$.

Notice that in any strongly consistent ordering, part $X_{ij}$ must be between parts $X_i^F$ and $X_i^T$. Moreover, if $x_i^F < x_i^T$,
then $x_{ij}^F < x_{ij}^T$, and conversely. In particular, in any valid vertex order, for each $i \in \{1, \dots, r\}$, either
$x_{ij}^F < x_{ij}^T$ for every $j \in \{1, \dots, s\}$ or
$x_{ij}^T < x_{ij}^F$ for every $j \in \{1, \dots, s\}$.

The \textsc{Partitioned \fpp{k}} instance will be such that if there is a precedence strongly consistent ordering for the vertices with respect to the given parts, then the assignment $x_i = (x_i^F < x_i^T)$ (that is, $x_i$ is true if $x_i^F$ precedes $x_i^T$ in such an ordering and $x_i$ is false otherwise) satisfies $\varphi$ in the context of \textsc{Not all equal 3-SAT} and, conversely, if there is a truth assignment satisfying $\varphi$ in that context, then there exists a strongly consistent ordering for the \textsc{Partitioned \fpp{k}} instance in which $x_i^F < x_i^T$ if $x_i$ is true and $x_i^T < x_i^F$ otherwise.

In what follows, if the $k$-th literal $\ell_{ij}$ of $\mC_j$ is the
variable $x_i$ (resp. $\neg x_i$), we denote by $O_{ij}$ the
ordered part $\{x_{ij}^F,x_{ij}^T\}$ (resp.
$\{x_{ij}^T,x_{ij}^F\}$).

Given a 2-vertex ordered part $C$, we denote by $C^1$ and $C^2$
the first and second elements of $C$. By $\pm C$, we denote
``either $C$ or $\bar C$''.

For each clause $\mC_j = \ell_{1j} \vee \ell_{2j} \vee \ell_{3j}$, we add
the 2-vertex ordered parts $Y_{1j}$, $Y_{2j}$, and $Y_{3j}$, and
the edges $(O_{1j}^2,Y_{1j}^1)$, $(O_{1j}^1,Y_{2j}^1)$,
$(O_{1j}^2,Y_{2j}^1)$, $(O_{2j}^1,Y_{1j}^2)$,
$(O_{2j}^1,Y_{1j}^1)$, $(O_{2j}^2,Y_{2j}^1)$,
$(O_{2j}^2,Y_{2j}^2)$, $(O_{2j}^1,Y_{3j}^1)$,
$(O_{2j}^1,Y_{3j}^2)$, $(O_{3j}^1,Y_{1j}^2)$,
$(O_{3j}^2,Y_{1j}^2)$, $(O_{3j}^1,Y_{2j}^2)$,
$(O_{3j}^2,Y_{2j}^2)$, $(O_{3j}^2,Y_{3j}^1)$,
$(O_{3j}^2,Y_{3j}^2)$. These edges ensure the following properties
in every strongly consistent ordering of the graph with respect to the defined partition.

\begin{enumerate}

\item\label{it:1} Since $(O_{1j}^2,Y_{1j}^1)$ is the only edge
between $O_{1j}$ and $Y_{1j}$, their only possible relative
positions are  $O_{1j}<Y_{1j}$ and its reverse $\bar Y_{1j}<\bar O_{1j}$.

\item\label{it:2} Since $(O_{1j}^1,Y_{2j}^1)$ and
$(O_{1j}^2,Y_{2j}^1)$ are the edges between $O_{1j}$ and $Y_{2j}$,
their possible relative positions are $\pm O_{1j}<Y_{2j}$ and
$\bar Y_{2j}<\pm O_{1j}$.

\item\label{it:3} Since $(O_{2j}^1,Y_{1j}^1)$ and
$(O_{2j}^1,Y_{1j}^2)$ are the edges between $O_{2j}$ and $Y_{1j}$,
their possible relative positions are $\bar O_{2j}<\pm Y_{1j}$
and $\pm Y_{1j}<O_{2j}$.

\item\label{it:4} Since $(O_{2j}^2,Y_{2j}^1)$ and
$(O_{2j}^2,Y_{2j}^2)$ are the edges between $O_{2j}$ and $Y_{2j}$,
their possible relative positions are $O_{2j}<\pm Y_{2j}$ and
$\pm Y_{2j}<\bar O_{2j}$.

\item\label{it:5} Since $(O_{2j}^1,Y_{3j}^1)$ and
$(O_{2j}^1,Y_{3j}^2)$ are the edges between $O_{2j}$ and $Y_{3j}$,
their possible relative positions are $\bar O_{2j}<\pm Y_{3j}$
and $\pm Y_{3j}<O_{2j}$.

\item\label{it:6} Since $(O_{3j}^1,Y_{1j}^2)$ and
$(O_{3j}^2,Y_{1j}^2)$ are the edges between $O_{3j}$ and $Y_{1j}$,
their possible relative positions are $\pm O_{3j}<\bar Y_{1j}$
and $Y_{1j}<\pm O_{3j}$.

\item\label{it:7} Since $(O_{3j}^1,Y_{2j}^2)$ and
$(O_{3j}^2,Y_{2j}^2)$ are the edges between $O_{3j}$ and $Y_{2j}$,
their possible relative positions are $\pm O_{3j}<\bar Y_{2j}$
and $Y_{2j}<\pm O_{3j}$.

\item\label{it:8} Since $(O_{3j}^2,Y_{3j}^1)$ and
$(O_{3j}^2,Y_{3j}^2)$ are the edges between $O_{3j}$ and $Y_{3j}$,
their possible relative positions are $O_{3j}<\pm Y_{3j}$ and
$\pm Y_{3j}<\bar O_{3j}$.
\end{enumerate}

Notice that, by items \ref{it:1} and \ref{it:6} (resp. \ref{it:2}
and \ref{it:7}), the vertices of $Y_{1j}$  (resp. $Y_{2j}$) are
forced to lie between those of $O_{1j}$ and those of $O_{3j}$. More
precisely, the possible valid orders are $O_{1j}<Y_{1j},Y_{2j}<\pm O_{3j}$ and their reverses $\pm
O_{3j}<\bar Y_{2j} , \bar Y_{1j}<\bar O_{1j}$.

By items \ref{it:3} and \ref{it:4}, the vertices of $O_{2j}$ are
forced to be between those of $Y_{1j}$ and those of $Y_{2j}$. More
precisely, the possible valid orders are $\pm Y_{1j}<O_{2j}<\pm Y_{2j}$ and their reverses $\pm Y_{2j}<\bar O_{2j}<\pm
Y_{1j}$.

By items \ref{it:3} and \ref{it:5}, the vertices of $Y_{1j}$ and
$Y_{3j}$ are forced to be on the same side with respect to the
vertices of $O_{2j}$, either $\bar O_{2j}<\pm Y_{1j},\pm Y_{3j}$ or $\pm Y_{1j}, \pm Y_{3j}<O_{2j}$.

Hence, taking also into account item \ref{it:8}, the possible
valid orders are
\begin{itemize}
\item $O_{1j}<Y_{1j}, \pm Y_{3j} < O_{2j}<Y_{2j}<\bar O_{3j}$

\item $O_{1j}<Y_{2j}<\bar O_{2j}<Y_{1j} , \pm
Y_{3j}<\bar O_{3j}$

\item $O_{1j}<Y_{2j}<\bar O_{2j}<Y_{1j}<O_{3j}<\pm
Y_{3j}$
\end{itemize}
and their reverses,

\begin{itemize}
\item $O_{3j}<\bar Y_{2j}<\bar O_{2j}<\bar Y_{1j}
, \pm Y_{3j}<\bar O_{1j}$

\item $O_{3j}<\bar Y_{1j} , \pm Y_{3j}<O_{2j}<\bar Y_{2j}<\bar O_{1j}$

\item $\pm Y_{3j}<\bar O_{3j}<\bar Y_{1j}<O_{2j}<\bar
Y_{2j}<\bar O_{1j}$
\end{itemize}
and will correspond to truth assignments that make true, respectively,

\begin{enumerate}[a)]
\item $\ell_{1j} \wedge \ell_{j2} \wedge \neg \ell_{3j}$

\item $\ell_{1j} \wedge \neg \ell_{2j} \wedge \neg \ell_{3j}$

\item $\ell_{1j} \wedge \neg \ell_{2j} \wedge \ell_{3j}$

\item $\neg \ell_{1j} \wedge \neg \ell_{2j} \wedge \ell_{3j}$

\item $\neg \ell_{1j} \wedge \ell_{2j} \wedge \ell_{3j}$

\item $\neg \ell_{1j} \wedge \ell_{2j} \wedge \neg \ell_{3j}$

\end{enumerate}

Suppose first that there is a precedence strongly consistent ordering of $V(G)$ with respect to its vertex partition. Define a truth assignment for variables $x_1, \dots, x_r$ as $x_i = (x_i^F < x_i^T)$, for $i \in \{1, \dots, r\}$. 

As observed above, if the value of $x_i$ is true (resp. false), then
for every $j \in \{1, \dots, s\}$, the part $X_{ij}$ is ordered
$x_{ij}^F$ $x_{ij}^T$ (resp. $x_{ij}^T$ $x_{ij}^F$). So, for each
clause ${\mathcal C}_j$, the part corresponding to its $k$-th
literal will be ordered as $O_{kj}$ if the literal is assigned true
and as $\bar O_{kj}$ if the literal is assigned false. Since for
each valid order of the vertices there exist $k, k' \in \{1,2,3\}$
such that the part corresponding to the $k$-th literal is ordered
$O_{kj}$ and the part corresponding to the $k'$-th literal is
ordered $\bar O_{k'j}$, the truth assignment satisfies the instance
$\varphi$ of \textsc{Not all equal 3-SAT}.


Suppose now that there is a truth assignment for variables $x_1,
\dots, x_r$ that satisfies the instance $\varphi$ of \textsc{Not
all equal 3-SAT}. Define the order of the vertices in the
following way. The first $r$ vertices are $\{x_i^F : x_i \mbox{ is
true}\} \cup \{x_i^T : x_i \mbox{ is false}\}$, and the last $r$
vertices are $\{x_i^T : x_i \mbox{ is true}\} \cup \{x_i^F : x_i
\mbox{ is false}\}$. 
Between these first and last $r$ vertices,
place all the parts $X_{ij}$, $Y_{1j}$,$Y_{2j}$ and $Y_{3j}$ associated with each clause $\mathcal{C}_j$,
$j=1,\dots,s$.
In particular, the parts $X_{ij}$, $Y_{1j}$, $Y_{2j}$ and $Y_{3j}$ are ordered accordingly to which of the conditions
(a)--(f) is satisfied. 
By the analysis above, this is a precedence strongly consistent ordering of the vertices of $G$, with respect to the defined parts. 
As an example, Figure~\ref{fig:ppthinnes:npc:instance} depicts the instance of the \textsc{Partitioned \fpp{k}} problem built from the instance $\varphi = \{(x_1\wedge x_2 \wedge x_3)\}$ of \textsc{Not all equal 3-SAT} problem.
\end{proof}

\begin{figure}[htbp]
\centering     
\includegraphics[width=80mm]{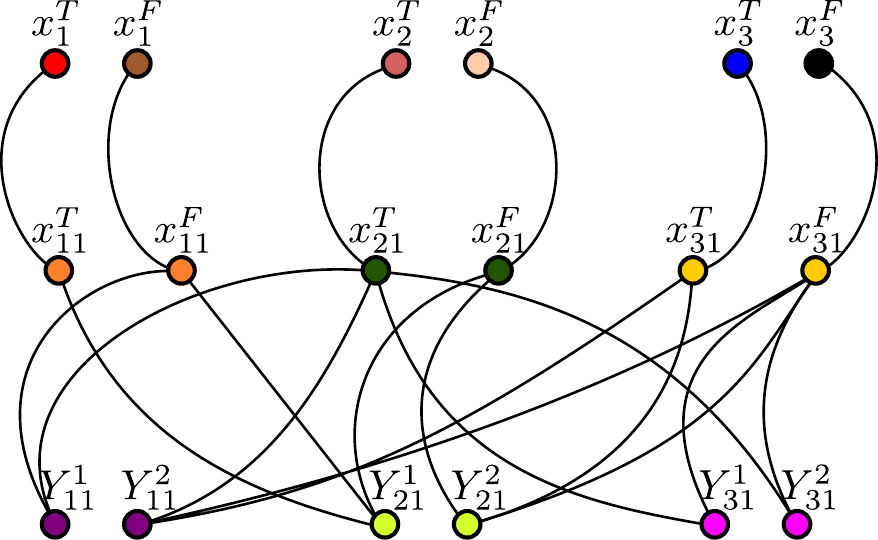}
\caption{ Instance of the \textsc{Partitioned \fpp{k}} problem built from the instance $\varphi = \{(x_1\wedge x_2 \wedge x_3)\}$ of \textsc{Not all equal 3-SAT} problem.}
\label{fig:ppthinnes:npc:instance}
\end{figure}

The remaining of this section is dedicated to discuss a polynomial time solution to a variation of the \textsc{Partitioned \fpp{k}} problem. This variation consists in considering a fixed number of parts for $V(G)$, that is, $k$ is removed from the input and  taken as a constant for the problem. The strategy that will be adopted is the same used for \textsc{Partitioned \fp{k}} problem. It is not difficult to see that Property~\ref{prop:thinness} is not sufficient to describe the requirements that must be imposed in the ordering of vertices in a precedence strongly consistent ordering.  This is so because, unlike what occurs in a precedence consistent ordering, in a precedence strongly consistent ordering the vertices of each part $V_i$ may impose an ordering to vertices that belong to parts that precede and succeed $V_i$. Given this fact, we observe the following property to describe such relation.

\begin{property}\label{prop:pthinness}
Let  $(V_1,V_2, \ldots, V_k)$ be a partition of $V(G)$, $s$ a precedence strongly consistent ordering and  $1 \leq i , j \leq k$. If $V_i$ precedes $V_j$ in $s$, then for all  $u,v \in V_i$ and $w \in V_j$, if $(u,w) \not\in E(G)$ and $(v,w) \in E(G)$, then  $u$ precedes $v$ in $s$. Moreover, for all  $u \in V_i$ and $w,x \in V_j$, if $(u,w) \not\in E(G)$ and $(u,x) \in E(G)$, then  $x$ precedes $w$ in $s$
\end{property}

\begin{figure}[htbp]
\centering     
\includegraphics[width=80mm]{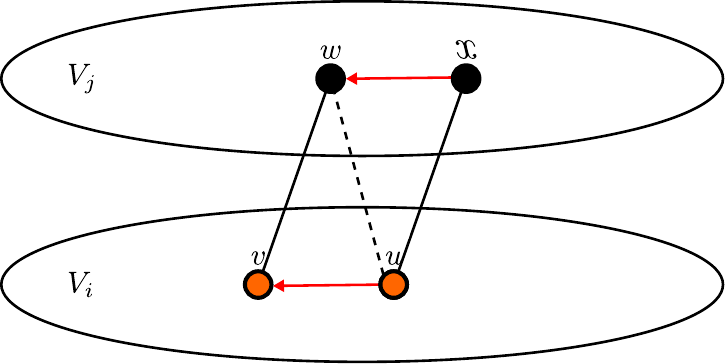}
\caption{Precedence relations among the vertices in a precedence strongly consistent ordering.}
\label{fig:propthinness}
\end{figure}

Notice that the greedy strategy used in Section~\ref{chpt:pthinness} does not work in the problem being considered. This is so because, according to Property~\ref{prop:pthinness} and visually depicted in Figure~\ref{fig:propthinness}, the ordering of vertices of $V_i$ in a precedence strongly consistent ordering $s$ is influenced by both the parts that precede and succeed $V_i$ in $s$. Despite this, the  method described in the Section~\ref{chpt:pthinness} to validate whether a part can precede a set of parts  is also useful to present a solution to this problem.  

Let $G$ be a graph, $\mathcal{V} =(V_1,V_2, \ldots, V_k)$ a partition of $V(G)$ and $s$ a precedence strongly consistent ordering of $V(G)$ for the given partition. 
Clearly, for all $1 \leq 1 \leq k$,  $G[V_i]$ must be a proper interval graph for $s$ to be a precedence strongly consistent ordering. Verifying whether $G[V_i]$ is a proper interval graph can be accomplished in linear time. If one of the parts does not induce a proper interval graph, then the answer is \textsc{NO}. Otherwise, each part has a $PQ$ tree associated to it. 

For a  given sequence $s_\mathcal{V}$ of parts of $\mathcal{V}$, suppose that $V_j<V_i<V_z$ in $s_\mathcal{V}$,  for $1 \leq j,i,z \leq k$. Let $T_i$ be a $PQ$ tree of $G[V_i]$.
Notice that, considering the Property~\ref{prop:pthinness},  if we apply Theorem~\ref{theo:thinness} to get the ordering constraints imposed by $V_j$ and $V_z$ to $T_i$, and add the directed edges  to $T_i$ in the same way that has been done in Section~\ref{chpt:pthinness} and $T_i$ can meet the constraints, then $T_i$ is compatible to being at that position. That is, the vertices of $V_i$ can precede the vertices of $V_z$ and succeed the vertices of $V_j$ in any precedence strongly consistent ordering. We show that for any $s_\mathcal{V}$, it is possible to verify whether there is a precedence strongly consistent ordering $s$ in which the ordering of the parts in $s$ is precisely $s_\mathcal{V}$.

To solve the problem, we will test all $k!$ possible permutations $s_\mathcal{V}= V'_1,V'_2,\ldots\allowbreak , V'_k$ among the parts of $\mathcal{V}$ and validate, using a digraph and $PQ$ trees, if each part $V'_i$ can precede  $V'_z$ and succeed $V'_j$, for all $1\leq j <i< z \leq k$. This validation is done exactly as described in Section~\ref{chpt:pthinness}, except for using Property~\ref{prop:pthinness} instead of Property~\ref{prop:thinness}. If there is some $s$ that satisfies this condition, then there is a precedence strongly consistent ordering with respect to $s$ and $G$ is a \fpp{k} graph concerning $\mathcal{V}$. Otherwise, $G$ is not a \fpp{k} graph with respect to $\mathcal{V}$.  Algorithm~\ref{alg:kfpp} formalizes the procedure.

\begin{algorithm}[htbp]
\textbf{Input}: $G$: a graph; 
        $\mathcal{V}$: a $k$-partition $(V_1, V_2, \ldots, V_k)$ of $V(G)$ for some fixed $k$  \\
\vspace{0.20cm}
\KpptF{
  
    \For{\Each  $V_i \in \mathcal{V}$ }{
         \If{\textsf{\upshape $G[V_i]$ is not a proper interval graph}}{
            \Return (\textsc{NO}, $\emptyset$)
         }
    }
    
    \For{\Each \textsf{\upshape  permutation $s_\mathcal{V}$ of  $\mathcal{V}$ }}{
        $s \leftarrow \emptyset$\\
        $foundValidPermutation \leftarrow$ \TRUE\\
        \For{\Each $V_i \in s_\mathcal{V}$}{
            Create a digraph $D = (V_i, \emptyset)$\\
            Build a $PQ$ tree $T_i$ of $G[V_i]$\\
            \For{\Each \textsf{\upshape $V_j \in s_\mathcal{V}$ such that $V_j \neq V_i$}}{
                Let $S$ be the set of precedence relations among the vertices of $V_i$ concerning $V_j$ (Property~\ref{prop:pthinness})\\
                \For{\Each $(u<v) \in S$}{
                   $E(D) \leftarrow  E(D) \cup \{(u,v)\}$ \\
                    \For{\textsf{\upshape \Each node $X$ of $T_i$}}{
                        Add the direct edges, deriving from $(u<v)$, among the children of $X$ (Theorem~\ref{theo:thinness})  \\
                    }
                }
                \For{\textsf{\upshape \Each node $X$ of $T_i$}}{
                    Let $D = (V', E')$ be the digraph where $V'$ is the set of the children of $X$ and  $E'$ are the directed edges added among them\\
                     \If{\textsf{\upshape there is a topological ordering $s$ of $D$}}{
                         Arrange the children of $X$ according to $s$ \\
                    }
                    \Else{
                        $foundValidPermutation \leftarrow$ \FALSE \\
                    }
                }
            }
            \If{ foundValidPermutation}{
                \AddEdgesFromPQTree{$G[V_i]$, $D$, $T_i$}\\
                \If{\textsf{\upshape there is a topological ordering $s_i$ of $D$}}{
                    $s \leftarrow ss_i$\\
                }
                \Else{$foundValidPermutation \leftarrow$ \FALSE \\
                \Break}
            }
        }
        \If{ foundValidPermutation}{
            \Return (\textsc{YES}, $s$)
         }
    }
    \Return (\textsc{NO}, $\emptyset$)
}
\caption{\textsc{Partitioned \fpp{k}}}
\label{alg:kfpp}
\end{algorithm}

Concerning the time complexity of the algorithm, first note that to create in $T_i$ the directed edges derived from Property~\ref{prop:pthinness} related to $V'_j$ (resp. $V'_z$) can be done in $O(n^5)$ time. Also, for each $V_i$ we apply this property considering all the other parts, that is, $O(k)$ times, and therefore $O(k^2)$ times overall considering each $V_i$. As this operation must be executed for all $k!$ possible permutations, and considering the analysis of this same method in Section~\ref{chpt:pthinness}, the given strategy yields a worst case time complexity of $O(k!k^2n^5) = O(n^5)$ as $k$ is fixed.

We end this section by mentioning an even more restricted case of the problem. Namely, the recognition of \fpp{k} graphs for a fixed number of parts such that each part induces a connected graph. Note that, as each part induces a connected graph, the proper interval graph induced by each part has an unique proper canonical ordering but reversion or mutual true twins permutation. This fact implies that the $PQ$ tree related to each one of these proper interval graphs is formed by one node of type $Q$, which is the root, that has all the maximal cliques as its children. That is, there are only two possible configurations for each one of these $PQ$ trees. As the number of possible configurations is a constant, this property leads to a more efficient algorithm. Instead of using Theorem~\ref{theo:thinness} to map restrictions to the $PQ$ tree in order to obtain a compatible tree, the algorithm can check both configurations of the $PQ$ tree independently. As the step which uses Theorem~\ref{theo:thinness} is no longer required, this approach leads to an  algorithm that yields a worst case time complexity of $O(k!k^22^kn^3)=O(n^3)$. This strategy is presented in Algorithm~\ref{alg:kfppc}.   

\begin{algorithm}[htbp]

\textbf{Input}: $G$: a graph; 
        $\mathcal{V}$: a $k$-partition $(V_1, V_2, \ldots, V_k)$ of $V(G)$, for some fixed $k$, such that $G[V_i]$ is connected for all $1 \leq i \leq k$  \\
\vspace{0.20cm}
\KpptFC{
  
    \For{\Each  $V_i \in \mathcal{V}$ }{
         \If{\textsf{\upshape $G[V_i]$ is not a proper interval graph}}{
            \Return (\textsc{NO}, $\emptyset$)
         }
    }
    
    \For{\Each  \textsf{\upshape permutation $s_\mathcal{V}$ of  $\mathcal{V}$ }}{
        $s \leftarrow \emptyset$\\
        $foundValidPermutation \leftarrow$ \TRUE\\
        \For{\Each $V_i \in s_\mathcal{V}$}{
            $foundValidTree \leftarrow$ \FALSE\\
            Build a $PQ$ tree $T_i$ of $G[V_i]$\\
            Let $T_i'$ be the $PQ$ tree obtained from $T_i$ by reversing the order of the children of the root\\
            \For{\Each $T \in \{T_i,T_i'\}$}{
                Create a digraph $D = (V_i, \emptyset)$\\
                \For{\Each $V_j \in s_\mathcal{V}$ such that $V_j \neq V_i$}{
                    Let $S$ be the set of precedence relations among the vertices of $V_i$ concerning $V_j$ (Property~\ref{prop:pthinness})\\
                    \For{\Each $(u<v) \in S$}{
                       $E(D) \leftarrow  E(D) \cup \{(u,v)\}$ \\
                    }
                }
                \AddEdgesFromPQTree{$G[V_i]$, $D$, $T$}\\
                \If{\textsf{\upshape there is a topological ordering $s_i$ of $D$}}{
                    $s \leftarrow ss_i$\\
                    $foundValidTree \leftarrow$ \TRUE\\
                    \Break
                }
            }
            \If{\Not foundValidTree}{
                $foundValidPermutation \leftarrow$ \FALSE\\
                \Break
            }
            
        }
        \If{ foundValidPermutation}{
            \Return (\textsc{YES}, $s$)
         }
    }
    \Return (\textsc{NO}, $\emptyset$)
}
\caption{\textsc{Partitioned \fpp{k}}}
\label{alg:kfppc}
\end{algorithm}
\unskip
\section{Characterization of \fp{k} and \fpp{k} Graphs} \label{chpt:characterization}
This section describes a characterization of \fp{k} and \fpp{k} graphs for a given partition. First, we define some further concepts.

A graph $G$ is a  \emph{split} graph if there is a bipartition $(V_1,V_2)$ of $V(G)$ such that $V_1$ is a clique and $V_2$ a stable set of $G$. A graph $G$ is called a \emph{threshold} graph if $G$ is a split graph and there is  an ordering of $V_1$ (resp. $V_2$), named \emph{threshold ordering}, such that the neighborhood of vertices of $V_1$ (resp. $V_2$) are ordered by inclusion, that is, if $u$ precedes $v$ in the threshold ordering, $N[u] \subseteq N[v]$ (resp. $N(u) \subseteq N(v)$).

For the following characterization, we will define the split graph $S_G(V_1,V_2)$ with respect to a bipartition $(V_1,V_2)$ of a graph $G$. Such a graph is obtained from $G$ by the completion of edges among the vertices of $V_1$ and the removal of all edges among the vertices of $V_2$, hence transforming $V_1$ into a clique and $V_2$ into a stable set.
The Figures~\ref{fig:b} and~\ref{fig:d} illustrate the corresponding split graphs of the graphs in the Figures~\ref{fig:a} and~\ref{fig:c}, respectively.

Let $s = s_1s_2$ be an ordering of $V(G)$. We define $(s_1,s_2)$ as \emph{in accordance} with $G$ if $s_1$ is a threshold ordering of  $S_G(V(s_1),V(s_2))$, and if  $s_1$ and $s_2$ are canonical orderings of $G[V(s_1)]$ and  $G[V(s_2)]$, respectively.
Additionally, we define $(s_1,s_2)$ as \emph{strongly in accordance} with $G$ if  $s_1$ and $s_2$ are proper canonical orderings of $G[V(s_1)]$ and  $G[V(s_2)]$, respectively, and both $s_1$ and $\bar{s_2}$ are threshold orderings of  $S_G(V(s_1),V(s_2))$.

As an example, let $s_1$ and $s_2$ be the orderings represented in the Figure~\ref{fig:abcd} by reading the vertices of each part, of each graph, from left to right. The related pair $(s_1,s_2)$ of  Figure~\ref{fig:a} is in accordance, but is not strongly in accordance, with the given graph. 
In the order hand, Figure~\ref{fig:c} depicts a pair $(s_1$,$s_2)$ which is strongly in accordance with the associated graph. Finally, Figure~\ref{fig:e} exemplifies a case where  the given orderings are neither in accordance or strongly in accordance with its correlated graph.
In fact,  since both $V_1$ and $V_2$ in Figure~\ref{fig:e} induce subgraphs that admit only four canonical orderings each, it can be easily verified that there is  no $(s_1, s_2)$ which is in accordance, or strongly in accordance, with the graph.

\begin{figure}[htb]
    \centering 
    \subfigure[$G_1$]{\label{fig:a}\includegraphics[width=30mm]{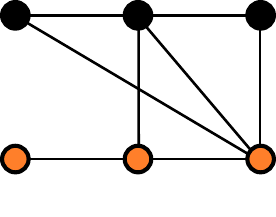}}\hfil
    \subfigure[$S_{G_1}(V_1,V_2)$]{\label{fig:b}\includegraphics[width=30mm]{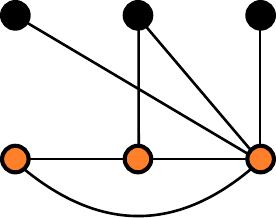}}\hfil \\
    \subfigure[$G_2$]{\label{fig:c}\includegraphics[width=30mm]{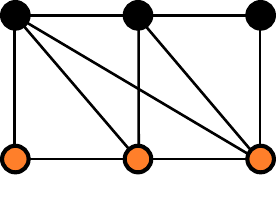}}\hfil
    \subfigure[$S_{G_2}(V_1,V_2)$]{\label{fig:d}\includegraphics[width=30mm]{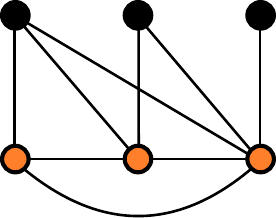}}\hfil \\
    \subfigure[$G_3$]{\label{fig:e}\includegraphics[width=30mm]{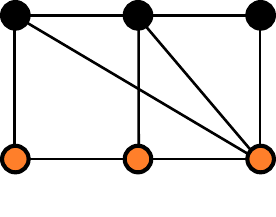}}\hfil
    \subfigure[$S_{G_3}(V_1,V_2)$]{\label{fig:f}\includegraphics[width=30mm]{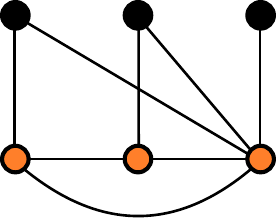}}

\caption{Corresponding split graphs ($V_1$ is the set of orange vertices and $V_2$ the black ones).}
\label{fig:abcd}
\end{figure}

\vspace{.2cm}
\begin{lemma} \label{lem:2pt}

Let $G$ be a graph. Then, $G$ is \fp{2} if, and only if, there is a consistent ordering $s = s_1s_2$ for which $\mathcal{V}=(V(s_1),V(s_2))$ is a bipartition of $V(G)$ 
such that $(s_1,s_2)$ is in accordance with $G$.

\end{lemma}

\begin{proof}
Let $\mathcal{V} = \{V_1, V_2\}$ be a bipartition of $V(G)$ and $s_1$ and $s_2$ be two total orderings of $V_1$ and  $V_2$, respectively.

Consider $G$ is \fp{2} concerning  $\mathcal{V}$. Let $s=s_1s_2$ be a precedence consistent ordering of $V(G)$. Thus, $s_1$ is a canonical ordering of $G[V_1]$.
Suppose by absurd that $s_1$ is not a threshold ordering of $S_G(V_1,V_2)$. 
That is, there are $u,v \in V_1$ and  $w \in V_2$, with $u<v$ in $s_1$, such that $w \in N[u]$ and $w \not\in N[v]$. As $u<v<w$ em $s$, there is a contradiction with the fact that $s$ is a precedence consistent ordering of $V(G)$. Therefore, 
$(s_1,s_2)$ is in accordance with $G(V_1,V_2)$.

On the other hand, consider that $(s_1,s_2)$ is in accordance with $G$. Thus, both $s_1$ and $s_2$ are canonical orderings of $G[V_1]$ and  $G[V_2]$, respectively, and  $s_1$ is a threshold ordering of $S_G(V_1,V_2)$. Now we prove that $s=s_1s_2$ is a precedence consistent ordering of $V(G)$ concerning  $\mathcal{V}$. Suppose by absurd that this statement does not hold. That is, there are  $u,v \in V_1$ and $w \in V_2$, with $u < v$ in $s$, such that $(u,w) \in E(G)$ and $(v,w) \notin E(G)$. This is a contradiction with the fact that $s_1$ is a threshold ordering  ordering of $S_G(V_1,V_2)$. Hence, $s=s_1s_2$ is a precedence consistent ordering of $V(G)$ concerning  $\mathcal{V}$.
\end{proof}

\begin{lemma}\label{lem:2ppt}
Let $G$ be a graph. Then, $G$ is \fpp{2} if, and only if, there is a strongly consistent ordering $s = s_1s_2$ for which $\mathcal{V}=(V(s_1),V(s_2))$ is a bipartition of $V(G)$ 
such that $(s_1,s_2)$ is strongly in accordance with $G$.

\end{lemma}

\begin{proof}
Let $\mathcal{V} = \{V_1, V_2\}$ be a bipartition of $V(G)$ and $s_1$ and $s_2$ be two total orderings of $V_1$ and  $V_2$, respectively.

Consider $G$ is  \fpp{2} concerning  $\mathcal{V}$. Let $s=s_1s_2$ be a precedence strongly consistent ordering of $V(G)$. Thus, $s_1$ is a proper canonical ordering of $G[V_1]$ and, as $G$ is also a \fp{2} graph, $s_1$ is a threshold ordering of $S_G(V_1,V_2)$, according Lemma~\ref{lem:2pt}.
Suppose by absurd that $\bar{s_2}$ is not a threshold ordering of $S_G(V_1,V_2)$. 
That is, there are $u,v \in V_2$ and $w \in V_1$, with $u<v$ in $s_2$, such that $w\in N(v)$ e $w\not\in N(u)$. That is, a contradiction with the fact that $s$ is a precedence strongly consistent ordering, as $w < u < v$ in $s$. Hence, 
$(s_1,s_2)$ is strongly in accordance with $G(V_1,V_2)$.

Now consider ($s_1$,$s_2$) is strongly in accordance with $G$. That is, both $s_1$ and $\bar{s_2}$ (resp. $s_1$ and $s_2$) are threshold orderings (proper canonical orderings) of $S_G(V_1,V_2)$ (resp. $G[V_1]$ and $G[V_2]$, respectively). 
By Lemma~\ref{lem:2pt}, $s=s_1s_2$ is a precedence consistent ordering of $V(G)$. 
Next, we prove that $s=s_1s_2$ is also a precedence strongly  consistent ordering of $V(G)$ concerning  $\mathcal{V}$. For the sake of contradiction, suppose that the statement does not hold. That is, there are  $u,v \in V_2$ and $w \in V_1$, with $u < v$ in $s$, such that $(v,w) \in E(G)$ and $(u,w) \notin E(G)$. This is an absurd, as $\bar{s_2}$ is a threshold ordering  ordering of $S_G(V_1,V_2)$. Consequently, $s=s_1s_2$ is a precedence strongly consistent ordering of $V(G)$ concerning  $\mathcal{V}$.
\end{proof}
 
The above lemmas can be generalized to an arbitrary number of parts as follows.
 
\begin{theorem}\label{teo:kppt}
Let $G$ be a graph. For all $k > 2$, $G$ is \fp{k} (resp. \fpp{k}) if, and only if, there is a precedence consistent (resp. strongly consistent) ordering $s = s_1 \ldots s_k$ for which $\mathcal{V}=(V(s_1),V(s_2)),\ldots,V(s_k))$ is a $k$-partition of $V(G)$  such that for all $1 \leq i<j \leq k$, $(s_i,s_j)$ is in accordance (resp. strongly in accordance) with  $G[V(s_i) \cup V(s_j)]$.
\end{theorem}

\begin{proof}

Suppose there are  a total ordering $s = s_1 \ldots s_k$   and a partition $\mathcal{V}=(V(s_1),V(s_2)),\ldots,V(s_k))$ of $V(G)$. Notice that $s$ is a precedence consistent (resp. strongly consistent) ordering if, and only if, for all $1 \leq i < j \leq k$, $s_i s_j$ is a precedence consistent (resp. strongly consistent) ordering of $G[V_i \cup V_j]$ concerning the bipartition $(V_i, V_j)$. By Lemma~\ref{lem:2pt} (resp. Lemma~\ref{lem:2ppt}), it holds if, and only if,  $(s_i,s_j)$ is in accordance (resp. strongly in accordance) with  $G[V(s_i) \cup V(s_j)]$.
\end{proof}
\unskip
\section{Conclusions and open problems}\label{chpt:future}

In this work, we study two classes of graphs: \fpExt{k} and \fppExt{k} graphs, subclasses of $k$-thin and proper $k$-thin graphs, respectively. 
Concerning \fpExt{k} graphs, we present a polynomial time algorithm that receives as input a graph $G$ and a $k$-partition of $V(G)$ and decides whether $G$ is a precedence $k$-thin graph with respect to the given partition. This result is presented in Section~\ref{chpt:pthinness}. Regarding \fppExt{k} graphs, for the same input, we prove that if $k$ is a fixed value, then it is possible to decide whether $G$ is a precedence proper $k$-thin graph with respect to the given partition in polynomial time. For variable $k$, the related recognition problem is \NP-complete. These results are presented in Section~\ref{chpt:ppthinness}. Also, using threshold graphs, we characterize both \fpExt{k} and \fppExt{k} graphs.  

Concerning the classes defined in this paper, some open questions are highlighted:

\begin{itemize}
    
    \item Given a graph $G$, what is the complexity of evaluating  \fpPar{G} and \fppPar{G}?
    
    \item Given a graph $G$ and an integer $k$, what is the complexity of determining  if \fpPar{G}, or \fppPar{G}, is at most $k$?
    
    \item How do \fpPar{G} and \fppPar{G} relate to $thin(G)$ and $pthin(G)$, respectively? 
    
    \item Is it possible to extend the results of this paper to consider other types of orderings (partial orders) and restrictions? 

\end{itemize}
 

\section*{Acknowledgments}
This work was partially supported by FAPERJ and CNPq, Brazil, and UBACyT, Argentina.
The third author is supported by CAPES, Brazil. We would like to thank V\'{\i}ctor Braberman for some ideas for the reduction in the \NP-completeness proof.


\end{document}